\documentclass[11pt]{article}

\usepackage{dsfont}             %
\usepackage{amsmath}
\usepackage{amsthm}

\usepackage{amssymb}
\usepackage{algorithmicx}
\usepackage{algorithm}
\usepackage{algpseudocode}
\usepackage{graphicx}
\usepackage{color}

\usepackage{hyperref}
\usepackage{mathabx}
\usepackage[letterpaper,margin=1in]{geometry}

\newcommand{\R}{\mathbb{R}}    
\newtheorem{theorem}{Theorem}[section]
\newtheorem{lemma}[theorem]{Lemma}
\newtheorem{definition}{Definition}
\newtheorem{proposition}[theorem]{Proposition}

\newtheorem{cor}{Corollary}
\newtheorem{fact}{Fact}
\newtheorem{claim}{Claim}

\theoremstyle{remark}

\newtheorem*{note}{Note}

\newcommand{\capac}{\eta}

\newcommand{\hF}{\hat{F}}
\newcommand{\tF}{\tilde{F}}
\newcommand{\hL}{\hat{L}}
\newcommand{\tL}{\tilde{L}}

\newcommand{\eps}{\varepsilon}
\newcommand{\opt}{\textrm{OPT}}

\newcommand{\dist}{\textrm{dist}}

\newcommand{\cost}{\textrm{cost}}

\newcommand{\calA}{\mathcal{A}}
\newcommand{\calE}{\mathcal{E}}

\newcommand{\calD}{\mathcal{D}}

\newcommand{\calC}{\mathcal{C}}
\newcommand{\calS}{\mathcal{S}}

\newcommand{\calF}{\mathcal{F}}

\newcommand{\calM}{\mathcal{M}}

\newcommand{\calP}{\mathcal{P}}

\newcommand{\local}{L}
\newcommand{\calI}{\mathcal{I}}
\newcommand{\globalS}{\mathcal{G}}
\newcommand{\globalSs}{\mathcal{G}^{*}}
\newcommand{\globalSh}{\widehat{\mathcal{G}}}

\newcommand{\epsval}{(\eps/(p+1))^{6p}}

\def\eg{\textit{e.g., }}
\def\ie{\textit{i.e., }}

\begin{document}

%%%%%%%%%%%%%%%%%%%%%%%%%%%%%%%%%%%%%%%%%%%%%%%%%%%%%%%%%%%%%%%%%%%%%%%%%%%%%%%%%
%
%                         TITLE
% 
%%%%%%%%%%%%%%%%%%%%%%%%%%%%%%%%%%%%%%%%%%%%%%%%%%%%%%%%%%%%%%%%%%%%%%%%%%%%%%%%%

\title{Approximation Schemes for Capacitated Clustering in
  Doubling Metrics}
\author{Vincent Cohen-Addad\\
  Sorbonne Universit\'e, UPMC Univ Paris 06, CNRS, LIP6}
\date{}

\maketitle
\begin{abstract}
  We consider the classic
  uniform capacitated $k$-median and uniform
  capacitated $k$-means problems
   in bounded doubling metrics.
  
  We provide the first QPTAS
  for both problems and the 
  first PTAS for the $k$-median version
  for points in $\R^2$.

  This is the first improvement over the bicriteria
  QPTAS for capacitated $k$-median in low-dimensional
  Euclidean space of Arora, Raghavan, Rao [STOC 1998]
  ($1+\eps$-approximation, $1+\eps$-capacity violation) and
  arguably the first polynomial-time approximation algorithm for
  a non-trivial metric.
  Our result relies on a new structural proposition that applies
  to any metric space and that may be of interest for developping
  approximation algorithms for the problem in other metric spaces,
  such as for example planar or minor-free metrics.
\end{abstract}

% \section{Dissection Procedure}
% \label{sec:diss}
% \input{struct}
% \section{A Structured Near-Optimal Solution}
% \label{sec:structthm}
% \input{approx}
% \section{Proof of Theorem~\ref{thm:main}}
% \label{sec:mainthm}
% \input{proofofthm}
% \section{A Dynamic Program to Find the Best Improvement}
% \label{sec:DP}
% \input{alg}

%% \section{Capacitated $k$-Median and $k$-Means}
%% \label{sec:capacitated}
%% \newpage
\section{Introduction}
The capacitated $k$-median and $k$-means problems are infamous problems:
no constant
factor approximation is known for any non-trivial metric, even when the
capacities are uniform. 
Given a set of points $\calC$ in a metric space together with an integer $\eta$, the capacitated clustering problem asks 
for a set $C$ of $k$ points, called \emph{centers}, together an assignment $\mu : \calC \mapsto
C$ that assigns at most $\eta$ clients to any cluster and
such that the sum of the $p$th power of the distance from each point to the center it is assigned
to is minimized (see a more formal
definition in Section~\ref{sec:ourresult}). When $p=1$, this is known as the \emph{capacitated $k$-median problem with uniform capacities},
while the case $p=2$ is the \emph{capacitated $k$-means problem with uniform capacities}.

The best known algorithm
is folklore and is an $O(\log k)$-approximation arising from Bartal's
embedding into trees
and a simple dynamic program for solving the problem exactly in
time $n^{O(t)}$ in graphs of treewidth
at most $t$ (in this case $t=1$).
From a theory perspective, finding a constant factor approximation
for the problem in general metric spaces or showing that none exists
unless P=NP is an important challenge that has received a lot of
attention (see for example the large amount of work on bicriteria
approximations or on the facility location version of the
problem~\cite{byrka2016approximation,DBLP:conf/soda/Li15,DBLP:conf/soda/Li16,DBLP:journals/talg/Li17,DBLP:conf/icalp/DemirciL16,byrka2014bi,Chuzhoy:2005:AKM:1070432.1070569,charikar2002constant},
and the recent work on approximation algorithm with running
$\exp(k)$ poly $(n)$~\cite{2018arXiv180905791A,abs-1901-04628,Cohen-AddadL19}).
The only known hardness of approximations bounds are constant and are the ones obtained for the uncapacitated
$k$-median, see~\cite{GuK99} for general metrics and~\cite{CK19} for $\ell_p$-metrics.
In fact, the capacitated $k$-median problem has been presented
as one of the most fundamental problems for which determining whether
$O(1)$-approximation is possible is still an open
problem~\cite{2018arXiv180905791A}.
This hardness seems to extend to any non-trivial metric (bounded
treewidth graphs excepted) since no constant factor approximation when the
input consists for example
of point in $\R^2$ is known. This stands in sharp contrast
with the uncapacitated variant of the problem for which
approximation schemes are known. 

Thus, since the breakthrough of
Arora et al.~\cite{ARR98} on
clustering problems in low-dimensional Euclidean space,
it has remained an important open problem to obtain at least a
constant factor approximation for capacitated
clustering problems even in $\R^2$.
Since their work, the community has developed two
main techniques for obtaining
approximation schemes for clustering problems in metrics
of fixed doubling dimension
or low-dimensional Euclidean space:
the approach of Kolliopoulos and Rao~\cite{KoR07}
and the local search algorithm (\cite{FRS16a,CAKM16,Cohen-Addad18}).
Unfortunately, the
approach of Kolliopoulos and Rao requires to reassign
clients among the optimal set of centers and so, cannot be adapted
to the case where centers have capacities (as also pointed out
by Sariel Har-Peled in
the comments of a StackExchange discussion~\cite{stack})
Furthermore, it is
easy to come up with a set of points in $\R^2$ where 
local search approach may have an arbitrarily bad approximation ratio.
Thus, the best algorithm for the problem in $\R^2$ is
the 20-year old bicriteria QPTAS of
%w.r.t. capacity constraint, $(1+\eps)$-approximation
Arora et al.~\cite{ARR98} (see again the discussion at~\cite{stack}).
Namely,
an algorithm that computes in
time $n^{\textrm{poly}(\eps^{-1}) \log^{O(1)} n}$ a 
solution that opens at most $k$ centers, that assigns up to
$(1+\eps)\eta$ clients to each center, and whose cost is at most
$(1+\eps)$ times the cost of the optimal solution that opens at most $k$
and assigns at most $\eta$ clients per cluster.

Arguably, the complexity comes from the current lack of techniques
for handling both the cardinality
constraint on the maximum number of centers in the solution, $k$, and
the hard capacity constraint
on the number of points that can be assigned to a center. Indeed, if
one of these two conditions
can be violated by some constant factor, then constant factor
approximation algorithms are known~\cite{byrka2016approximation,DBLP:conf/soda/Li15,DBLP:conf/soda/Li16,DBLP:journals/talg/Li17,DBLP:conf/icalp/DemirciL16,byrka2014bi,Chuzhoy:2005:AKM:1070432.1070569,charikar2002constant}.
Unfortunately, state-of-the-art algorithms producing
a $(1+\eps)$-violation of the capacity ask to sacrifice on both
the approximation guarantee and, perhaps more importantly on the
running time: the best known algorithms achieve a
$O(1/\eps^2)$-approximation in time $n^{1/\eps}$. In comparison, with
the same amount of time, our algorithm provides a solution that meets
the capacities exactly and that is a $(1+\eps)$-approximation for $\R^2$.

%, one of them arising from redistricting.

\paragraph{Applications in $\R^2$ and $\R^{O(1)}$.}
The capacitated clustering problem has received a lot of attention
through the years. The study of the problem in low-dimensional
metric stem from \emph{prepositioning resources} and
\emph{redistricting}.
For example, consider the problem of positioning emergency
supplies to support disaster relief by dispatching a certain number of
emergency items, such as medicine, food, centers on a map.
There has been
a large body of work in this direction and it has been argued that
the capacity constraint of the centers is often a hard constraint,
in particular when it comes to medical centers where the cost of
service increases very quickly when the capacity is
violated~\cite{farlow2011prepositioning}
see also~\cite{kinay2018modeling,rawls2010pre,klibi2018,barahona1998plant}
for more discussions on hard capacities for prepositioning emergency
resources and hard capacity constraints
in plant locations.

A slightly less dramatic motivation also comes
from bike-sharing systems: the goal is place bike-stations
(i.e.: facilities) so as to cover a certain demand of people (users)
so that, at the end of its trip, a user can find a spot at a station
to leave its bike: the number of spots at a given station is then
a hard constraint.

The redistricting problem is the
problem of dividing a region into a number of electoral districts under
some hard constraints, often coming from the constitution of
the country. The first of the hard constraints is the number of districts,
which is our number of \emph{clusters} $k$, and which is, in many
country like France or the US, fixed
by law for a given region. Thus, computing a redistricting into
$(1+\eps)k$ districts is not an option. The second of the hard constraints
is the size of the districts. In the US for example, even though
the Supreme Court has declined to
name a specific percentage limit on how much populations of
districts can differ, we observe from~\cite[p.~499]{FryerHolden} that
``\textit{a 2002 Pennsylvania redistricting plan was
struck down because one district had... 19 more people than
another}''. It follows that since $\eta$ is a few thousands for these
instances, having a capacity violation of $(1+\eps) \eta$ would not be
satisfactory, unless $\eps$ could be made very tiny.
For example, some states accept a $1.5\%$ difference in size, but
to obtain such a small imbalance in sizes, current techniques
would require a running time of $n^{60}$ at best, while producing a
poor approximate solution.

%% We point out that here it is critical that the capacities
%% of the clusters are the same and so $\eta = n/k$.
Finally, as shown experimentally by~\cite{FryerHolden},
the $k$-means or $k$-median
objectives are suitable objective functions for evaluating the quality
of a solution (or what is referred to as its \emph{compactness},
see also \cite{Levitt}). The idea behind this being that, for the
Euclidean plane, for a given set of centers, the best assignment of
clients to centers under the capacity constraint can be phrased
as a Voronoi diagram in $\R^3$~\cite{DBLP:journals/corr/abs-1710-03358}.
Hence, the districts are convex, an appealing property.
In most of these works, assuming
that the input points are points in $\R^2$ is fairly standard
assumption (see~\cite{FryerHolden,DBLP:journals/corr/abs-1710-03358}
and references therein).

Therefore, designing good approximation algorithms for the
capacitated $k$-median and $k$-means problems in $\R^2$ and more
generally in metric spaces of fixed doubling dimension
has become an important challenge.

\subsection{Our Results}
\label{sec:ourresult}
We give the first PTAS for $k$-median with uniform capacities in $\R^2$
and the first
QPTAS for $k$-median, $k$-means with uniform capacities
in metrics of doubling dimension. The problem at hand is the following
\begin{definition}
  Let $\mathcal{X} = (X, \dist)$ be a metric space.
  Given a set of \emph{clients} $\calC \subseteq X$ in a metric space, a
  set of \emph{candidate centers} $\calA \subseteq X$, a capacity $\eta$,
  and $p\ge 1$,
  the \emph{capacitated $k$-clustering} problem asks for set
  $C \subseteq \calA$ of size at most $k$ and an assignement
  $\mu : \calC \mapsto C$ such that
  \begin{itemize}
  \item for any $f \in C$, $|\{c \mid c \in \calC \textrm{ and } \mu(c) = f\}| \le \eta$, and
  \item $\sum_{c \in \calC} \dist(c, \mu(c))^p$ is minimized.
  \end{itemize}
\end{definition}
Given a solution $(C, \mu)$, we refer to the candidate centers
in $C$ as \emph{centers} or \emph{facilities}. We say that $f \in C$
\emph{serves} a client $c$ if $\mu(c) = f$.

Our result in $\R^2$ is as follows.
It holds for any $\ell_p$-metric in $\R^2$, where $p = O(1)$.

\begin{theorem}
 \label{thm:ptas}
 There exists an algorithm that given an instance of size $n$ 
 of the capacitated $k$-median problem in $\R^2$ (capacitated
 clustering problem with $p=1$)        
 outputs     
 a $(1+\eps)$-approximate solution in time $n^{1/\eps^{O(1)}}$.
\end{theorem}

Our result extends to metric of bounded doubling dimension.
The \emph{doubling dimension} of a metric is the
smallest integer $d$ such that any
ball of radius $2r$ can be covered by $2^d$ balls of radius~$r$.
The result is as follows.

\begin{theorem}
\label{thm:qptas}
 There exists an algorithm that given an instance of size $n$ 
 of the capacitated $k$-clustering problem
 with parameter $p$
 in a metric space of
 doubling dimension $d$ outputs
 a $(1+\eps)$-approximate solution in
 time $n^{((\frac{p}{\eps})^{p}\log n)^{O(d)}}$.
\end{theorem}

\subsection{Techniques}
%% We have to overcome to main challenges. The first one is to find a solution
%% that meets the capacity constraint and the second one is to address not only
%% the simple $k$-median objective but to extend it to incorporate the $k$-means
%% objective as well. As previously discussed, the previous approach
%% of~\cite{ARR98}
%% only obtained a bicriteria QPTAS, thus meeting the right capacity and making
%% it a PTAS is already a challenge. 
Our main technical contribution and the meat of the paper
is Proposition~\ref{prop:main} which, interestingly,
holds in any metric space and so could perhaps be of use for solving
the open problem of getting an $O(1)$-approximation for the problem
in general metric spaces, or maybe more likely to obtain similar
in other contexts like for planar graph inputs.
We first provide some intuition on how we use Proposition~\ref{prop:main},
before describing what the proposition gives us.

%% As discussed in previous works, for example~\cite{Cohen-Addad18},
%% the classic use of the
%% quad-tree dissection or the split-tree decomposition of
%% Arora~\cite{ARR98} and
%% Talwar~\cite{Talwar04} does not work for the $k$-means objective
%% (or the $k$-clustering problem for $p>1$).

To understand better our contribution, we quickly review the
classic quad-tree dissection and split-tree decomposition techniques of
Arora~\cite{ARR98} and
Talwar~\cite{Talwar04}.
The general ideas behind the decomposition consists
in recursively partitioning the input into regions
and forcing the optimal solution to connect points
in different regions through a set
of \emph{portals} of size say $\rho^d$.
By doing so, one obtains small-size interfaces between regions
that enables dynamic programming techniques.
%% that could be enumerated in polynomial or quasi-polynomial time and
%% so.

Concretely, the technique ensures that
for a client in a given region $R$ that is assigned
to a center outside the region, the detour paid to connect the client to its center through
the set of portals
is $1/\rho$ times the diameter of $R$. The crux of the analysis is to show
that the probability that a client $u$ and a facility $v$ at distance
$\delta$ are in different
clusters of diameter $D$ is roughly $\delta/D$.
It follows that the expected
detour becomes $(D/\rho) \cdot \delta /D$ and so at most $1/\rho$ times
the distance between $u$ and $v$ in the original metric. Since this is
the original cost of serving $u$ by $v$, the cost of the optimal
solution that is forced to go through the portal is at most
$(1+1/\rho)$ higher than the optimal solution that does not have to
satisfy this constraint.
This works fine when the distance is equal to the cost
(namely when $p=1$).
% Here $\rho^{O(d)}$ is the number of portals used in the region.

However, for $p=2$ or larger the expected
cost of the detour now becomes $(D/\rho)^2 \cdot \delta/D =
 \delta D /\rho^2$. On the other hand, 
 the original cost of serving $u$ by $v$ is $d^2$ and so,
 when $D= \omega(d)$, the detour incurred by going to the closest
 portal may be too expensive.
This is one of the reasons why no PTAS was known for the uncapacitated $k$-means problem
until the work of~\cite{CAKM16,FRS16a}. Unfortunately, the algorithm of~\cite{CAKM16,FRS16a}
is local search and it is
easy to come-up with an instance where local search can have
arbitarily bad approximation
ratio for the capacitated version of the problem.

Our technique circumvents the above problem as follows.
Observe that in the above discussion, for any
client $u$ and the facility $f$
that serves $u$ in the optimal solution, if the regions
that contain $u$ and do not contain $f$ have diameter
at most $(\log n) \dist(u,f)/\eps$, then one can use a portal set of size
$\rho = (\log n/\eps)^{O(d)}$ and guarantee that the detour paid is in total
at most $\eps \dist(u,f)^2 $ which is $\eps$ times the cost paid by
$u$ in the solution.

Thus, we only have to worry about pairs of clients and
facilities for which the decomposition does not
provide such a nice structure. This leads us to say that a facility
or a client $p$ is
``badly cut'' (see formal definition in the next section)
which, at the intuitive level goes as follows.
A point is ``badly cut'' if, at some point in the decomposition,
there exists a region of diameter $D$
that contains $p$ but that does not contain some point that is at distance
$D/\textrm{poly} \log n$ from $p$.
In other words: $p$ is very close (relatively to $D$)
to the boundary of the region of diameter $D$.

As we argued before, for any point $p$ that is not badly cut,
we are in good shape,
we can afford to connect $p$ to the facility that serves it
$\opt$ through the portal.
The question is what to do with badly cut points. First, we will show
that a point is badly cut only with a tiny probability, say $\eps$
for a constant $\eps$.
For example,
if we were interested in a solution opening up to $(1+O(\eps))k$ centers,
we would almost be done: each facility of $\opt$ has probability
at most $\eps$ of being badly cut so we would have
at most $\eps k$ badly cut facilities in  expectation.
Therefore, we could
simply decide to consider a solution opening $2^{O(d)}$ facilities instead
of one for each badly cut facility: for the region where the facility
is badly cut, open one facility on each child region. Then we
would be guaranteed that no client is separated from its facility
at a too high level\footnote{One has to guarantee that this holds
  recursively, so the argument would be slightly more involved}. We could
then use dynamic programming to find such a solution.

However, our goal is to satisfy all constraints: at most $k$ facilities
open and at most $\capac$ points assigned to each facility.
To handle this, we use our main result, Proposition~\ref{prop:main}
which helps us deal with the badly cut clients and facilities and prove
that there exists a near-optimal solution that can be found through
a dynamic program.

The approach is as follows, we compute a
$\gamma$-approximate solution $L$ to the problem.
For any point $p$ that is badly cut,
we move $p$ to the location of the center serving $p$ in solution $L$.
This yields a new instance where any solution to the new instance can be
lifted back to the original instance by paying an extra
$O(\eps \cost(L))$ in expectation. We then need to argue that one can
find a $(1+\eps)$-approximation in the new instance.
%% This yields a new instance $\calI_{\calD}$.
To do so, we show that there exists a near-optimal solution that is well-behaved
in the new instance: we show 
show that there exists a solution of cost at most
$\cost(\opt) + O(\eps \cost(L))$ that contains each badly cut facility of $L$.
At first, this may seem unrealistic
since we want to end up with a
$(1+\eps)$-approximation while
still opening at most $k$ centers and preserving the capacity constraints.
%% forcing to open some centers seems a hard constraint, especially since we also
%% want to meet the capacity constraint.
This is where Proposition~\ref{prop:main} comes into place. %The proposition shows that
\paragraph{Main Result: Proposition~\ref{prop:main}}
As we have described the major ingredient is Proposition~\ref{prop:main}. Loosely
speaking, Proposition~\ref{prop:main} states that given a solution $(C, \mu)$ of
cost $X$ and a random process which picks each center of $C$ with probability $\eps^2$,
then with probability at least $1-\eps$,
there exists a solution which contains the selected centers and that:
(1) meets the capacity constraints (2) has at most $k$ centers, and (3) that is of cost
at most $\cost(\opt) + O(\eps X)$.

The result is obtained by designing a careful rerouting scheme of the clients,
involving min-cost max-flow techniques. While the result for uncapacitated
version of the problem can be obtained through a simple lemma; obtaining the same
bound for the capacitated case is much more challenging.

%% if each center of $L$ is badly cut with probability at most $\eps^2$, forcing
%% these centers still allows to obtain a solution of cost
%% at most $(1+\eps) \cost(\opt) + \eps \cost(L)$ (see Lemma~\ref{lem:valid}).

\bigskip
We then make use of the proposition as follows.
This provides us with a very good instance where (1) clients that are badly cut
are moved to the facility that serves them in $L$ and (2) badly cut facilities of
$L$ are now part of the solution we are trying to compute.
This is enough to conclude:
Consider a client $c$. If it
is not badly cut, then we don't have to worry about paying the detour through portals.
If it
is badly cut, then it is now located at the center that serves it in $L$.
Moreover, if this center is badly cut then it is open and so the service cost
for this client is 0. If this center is not badly cut, then one can afford
to make the detour to connect the client to its closest facility through the portals.
Making this reasoning rigorous is a bit challenging and shown in the next
sections.

%% An extra step is to show that moving a badly cut client to the facility that serves
%% it in $L$ still yields an instance that has small \emph{distortion} compared
%% to the original instance. Namely that any nearly-optimal solution in the new
%% instance results in a nearly-optimal solution in the original instance.

%% The final ingredient of our analysis is to design an efficient dynamic program
%% for capacitated $k$-median in $\R^2$. To do so we show how to drastically reduce
%% the number of DP states from $n^{O_{\eps}(\log n)}$ to $2^{O_{\eps}(\log n)}$.

\bigskip
A few more details still have to be addressed. 
Another problem we have to solve for making the entire approach work is the following.
Note that the solution we obtain has cost
at most $(1+\eps)\cost(\opt) + \eps(\cost(L))$ with probability at least
$1-\eps$. This probability can be boosted and we indeed boost it to
$1-\eps/\log \log n$ by repeating
$\log n$ times. 
This is critical since as discussed in the intro there is no
$O(1)$-approximation algorithm and so, the solution computed
has cost at most $(1+\eps) \opt + \eps \cost(L)$ which is 
$\eps \log n \cdot \cost(\opt)$.
However, this is enough to allow us to bootstrap:
we use the solution obtained to get a solution of cost at most
$\eps^2 \log n \cdot \cost(\opt)$.
By repeating  this process $\log \log n$ times, we finally obtain
%% a constant factor solution and can finally obtain in the next round
a near-optimal solution.

\subsection*{Organization of the Paper}
We provide definitions and preliminaries in the remainder of this section.
Our structural result, Proposition~\ref{prop:main} is presented in
Section~\ref{sec:struct}. To motivate the proposition, we first show how it is
used in Section~\ref{sec:badly} (Lemma~\ref{lem:valid}). From there a simple QPTAS follows,
see Section~\ref{sec:qptas} and a more involved PTAS is presented in Section~\ref{sec:ptas}.

\section{Preliminaries}
%% \label{sec:probdef}
%% We consider the following problem
%% \begin{definition}
%%   Let $\mathcal{X} = (X, \dist)$ be a metric space.
%%   Given a set of \emph{clients} $\calA \subseteq X$ in a metric space, a
%%   set of \emph{candidate centers} $\calF \subseteq X$, a capacity $\eta$,
%%   and $p\ge 1$,
%%   the \emph{capacitated $k$-clustering} problem asks for set
%%   $C \subseteq \calF$ of size at most $k$ and an assignement
%%   $\mu : \calA \mapsto C$ such that
%%   \begin{itemize}
%%   \item for any $f \in C$, $|\{c \mid c \in A \text{ and } \mu(c) = f\}| \le \eta$, and
%%   \item $\sum_{c \in \calA} \dist(c, \mu(c))^p$ is minimized.
%%   \end{itemize}
%% \end{definition}
Throughout the following sections, let $\eps > 0$.
Moreover, we will assume that $k \eta \ge n$ since otherwise, the problem has
no solution. %% Given a solution $(C,\mu)$ to the $k$-clustering
%% problem, we refer to the elements of $C$ as
%% \emph{centers} or \emph{facilities}.
The following observation and preprocessing step are folklore. Given a set of centers $C$, the assignment
$\mu$ minimizing the cost can be computed using a min cost flow algorithm by
defining a sink with capacity $\eta$ for each center of $C$, placing a demand
of 1 on each client, and for each client $c$ and center $f \in C$, defining an edge
$(c,f)$ with capacity one and cost $\dist(c,f)^p$.
Thus, given a set of centers the best assignment can always be computed in polynomial
time.

The following lemma will be useful to derive our results when $p>1$.
\begin{lemma}[E.g.: \cite{Cohen-AddadS17}]
  \label{sometriangleineq}
  Given a metric space $(X,\dist)$, and
  $p \ge 0$ and $1/2 > \eps>0$, we have for any $a,b,c \in X$,
  we have
  $\dist(a,b)^p \le (1+\eps)^p \dist(a,c)^p + \dist(c,b)^p(1+1/\eps)^p$.
\end{lemma}

\subsection{Doubling Metric Spaces and Decompositions}

Without loss of generality, we can assume that the aspect-ratio of
the input, namely the ratio of the maximum distance between pairs of
points in $\calA \cup \calC$ to the minimum distance between pairs
of points in $\calA \cup \calC$,
is at most $O(n^4)$.
Indeed, consider the following preprocessing step: compute
an $O(\log n)$-approximation and let $v$ be the
cost of the solution computed.
Then, while there is a pair of point $x,y$ that are at distance less than
$\eps v/n^4$, remove $x$ and add a client at $y$
(note that the two clients at $y$
may not necessarily be assigned to the same facility in a solution).
In the instance obtained at the end of this process, a point
is at distance
at most $\eps v/n^3$ from its original location, and so
any solution for this
instance can be converted back to the original instance
with a cost increase of at most
$\eps v/n^2 \le \eps \cost(\opt)$.
Finally, also note that, up to dividing all distance by the minimum
distance between any pair of points of the input, we can assume that
the minimum distance is 1 and so the diameter of the pointset, namely
the maximum distance between any pair of points, is equal to the
aspect-ratio.

A $\delta$-\emph{net} of $X$ is a set of point $Y$ such that $\forall v \in X,~
\exists x \in Y~|~ d(v, x) \leq \delta$ and $\forall x, y \in Y,~d(x, y) > 
\delta$. 
The cardinality of a 
net in metrics of doubling dimension $d$ is bounded by the following lemma.
\begin{lemma}[\cite{gupta2003bounded}]\label{prop:doub:net}
 Let $(X, \dist)$ by a metric space with doubling dimension $d$ and diameter 
$\Delta$, and let $Y$ be a $\delta$-net of $X$. Then $|Y| \leq 2^{d \cdot 
\lceil \log (\Delta/\delta)\rceil}$
\end{lemma}

%% We follow a few notations of~\cite{Cohen-Addad18}.
We define the \emph{rings} centered at a point $c$ as follows:
the $i$th \emph{ring} centered at  $c \in \calA \cup \calC$
is the set of all points
at distance $(2^i, 2^{i+1}]$ from $c$. The rings of $c$ is the collection
  of all the rings centered at $c$ which contains at least one point
  of the input.
  The following fact follows from the definition and
having aspect-ratio bounded by $O(n^4/\eps)$.
\begin{fact}
  The number of rings for any point $c$ is at most
  $O(\log(n/\eps))$.
\end{fact}

We use the \emph{randomized split-trees} of Talwar~\cite{Talwar04} for
doubling metrics and the randomized dissection of
Arora~\cite{arora1998polynomial}
for the plane. Since random split-trees are standard tools, we point
to~\cite{Talwar04} for a more detailed introduction. 
We use the exact same definition as~\cite{Talwar04}, with a slight
change in notation; We avoid talking about clusters but use the name
\emph{boxes} instead.
A decomposition of the metric $X$ is a partitioning of $X$ into
subsets, which we call boxes. A hierarchical decomposition
is a sequence of decompositions $\calP_0, \calP_1, \ldots , \calP_{\ell}$
such that every $\calP_i$ is a refinement of $\calP_{i+1}$, namely
each box of $\calP_i$ is a subset of a box of $\calP_{i+1}$.
The boxes of $\calP_i$ are the \emph{level-$i$} boxes.
A split-tree decomposition will
be one where $P_{\ell} = \{X\}$ and $P_0 = \{ \{x\} \mid x \in X\}$.
%% A split-tree of the metric stems naturally from a complete
%% hierarchical decomposition - the root node corresponds to
%% the single level l cluster X and the children of a level i cluster
%% are the level (i − 1) clusters comprising it. The leaves are
%% the singletons.

For any point $p$ and $x>0$, we say that the ball $B(p,x)$ is \emph{cut
at level $i$}, if there are $P_1,P_2 \in \calP_i$, and
$P_3 \in \calP_{i+1}$  such that
$P_1 \cap B(p,x) \neq \emptyset$ and $P_2 \cap B(p,x) \neq \emptyset$ and
$B(p,x) \subseteq P_3 \in \calP_{i+1}$.

%% The \emph{levels} of the boxes in the decomposition are as follows.
We obtain, a decomposition achieving the following properties
(see~\cite{Talwar04}):
%,Cohen-Addad18,ARR98,arora1998polynomial}):
\begin{enumerate}
\item The total number of levels $\ell$ is
  $O(\log n)$ (since the aspect ratio of the input metric is
  $O(n^4/\eps)$, and $\eps$ is constant).
\item Each box of level $i$ has diameter at most $2^{i+1}$, namely
  the maximum pairwise distances between points in a box of level
  $i$ is at most $2^{i+1}$.
\item Each box of level $i$ is the union of at most
  $2^{O(d)}$ level $i-1$ boxes.
\item\label{cond:rand}
  For any point $u$, $x>0$ and level $i$, the probability that the
  ball $B(u, x)$ is cut at level $i$ is
  $O(d \cdot x/2^i)$.

  %% For any pair of points $u,v$, the probability of
  %% $u$ and $v$ fall in different clusters at level $i$ is
  %% at most $O(d \cdot \dist(u,v)/2^i)$.
\end{enumerate}
Condition~\ref{cond:rand} is a direct corollary of the definition of
the decomposition and not stated precisely like this in~\cite{Talwar04}
but is fairly standard, see \eg \cite{abs-1812-08664} for a proof.

%% Now, fix a cluster $c$.
%% We apply Condition~\ref{cond:rand} to obtain that
%% the probability that all the clients in the $j$th ring
%% and $c$ fall in different boxes at level
%% $i$ is at most $O(k 2^j/2^i)$.

%% \begin{lemma}
%%   \label{lemma:probbasic}
%%   For any cluster $c$,
%%   the probability that there exists a client $x$ of the $j$th ring
%%   such that $x$ and $c$ fall in different balls at level
%%   $i$ is at most $O(k 2^j/2^i)$.
%% \end{lemma}
%% \begin{proof}
%%   Let $\eta = \dist(c,\omega)$, where $\omega$ is one of the
%%   centers of the balls drawn at level $i$.
%%   For the $j$th ring to be cut, we need
%%   to have that the $B(\omega, \rho 2^i)$ splits the
%%   ball $B(c, 2^{j+1})$. This means that $\rho 2^i$ falls
%%   in the interval $[\dist(c, \omega)-2^j, \dist(c, \omega)+2^j]$.
%%   Since the length of the interval is $2^{j+1}$,
%%   this happens with probability
%%   $2^{j+1}/2^{i}$. Summing over all the centers for
%%   the balls yields the result.
%% \end{proof}

For any point $c$,
for any ring $j$ of $c$, we say that it suffers a
\emph{bad cut} if the ball $B(c,2^j)$ is cut at a level $i$
higher than $\log(d (\log n/\epsval)) + j$,
namely $2^i > d (\log n/\epsval) 2^j$.
%% s.t. $2^i > (k \cdot \log n/\epsval) 2^j$.
We have:
\begin{lemma}
  \label{lem:ring:badcut}
  For any point $p$, the probability that a ring $j$
  centered at $p$ suffers a bad
  cut is at most $O(\epsval/\log n)$.
\end{lemma}
\begin{proof}
  By Condition~\ref{cond:rand},
  the probability to be cut at level
  $\log(d (\log n/\epsval)) + j + i$
  is $O(\epsval/(2^i\log n))$.
  Then, taking a union bound over all levels higher than
  $\log (d (\log n/\epsval)) + j$: we have
  %% Taking a union bound over all
  %% levels higher than $\log (d (\log n/\epsval)) + j$,
  %% yields
  that the total probability of suffering a
  bad cut is at most $O(\epsval/\log n) \sum_{i = 1}^{O(\log n)} 1/2^i$ 
  and so at most
  $O(\epsval/\log n)$.

\end{proof}

\section{Decompositions of Clustering Instances}
\label{sec:badly}
The goal of this section is to show how to use split-tree decompositions
and our structural result (Proposition~\ref{prop:main}) so as to
use a dynamic program to solve capacitated clustering instances.

Consider a metric space $(X,\dist)$ together with an instance
$\calC, \calA, \eta, p$ of the capacitated $k$-clustering problem
on $(X,\dist)$. Assume that $X$ has been preprocessed so as to have
an aspect-ratio of at most $O(n^4/\eps)$ as described in the previous
section and by removing points of $X$ that are not in $\calC \cup \calA$.
Let $\calD$ be a split-tree decomposition of $(X,\dist)$.
Let $L$ be a solution to the problem and $\opt$
denote an optimal solution.
For each point $c \in L \cup \opt \cup \calC$, we say that $c$
is \emph{badly cut} if at least one of its rings suffers a bad cut.

\begin{lemma}
  \label{lem:badlycut}
  For a given point $c \in L \cup \opt \cup \calC$,
  the probability that $c$ is badly cut is $O(\epsval)$.
\end{lemma}
\begin{proof}
  By
  Lemma~\ref{lem:ring:badcut} each ring suffers a bad cut with
  probability $O(\epsval/\log n)$.
  Thus, since the number of 
  rings is $O(\log (n/\eps))$, taking a union bound
  on the probability of each ring suffering a bad cut implies the lemma.
\end{proof}

Given a solution $(L, \mu_L)$ to $\calC, \calA, \eta, p$ and a random decomposition
$\calD$ of $(X, \dist)$, we define a new instance of
the $k$-clustering problem in metric $(X,\dist)$ as follows.
Any client $c$ that is badly cut is ``moved'' to the facility $L(c)$
that serves it in solution $L$. Namely, given the set of badly cut clients
$B_{\calD}$, the new instance is defined by the following tuple
$\calC - B_{\calD} \cup \{L(c) \mid c \in B_{\calD} \}, \calA, \eta, p$,
where $\calC - B_{\calD} \cup \{L(c) \mid c \in B_{\calD} \}$ is a multiset.
We identify clients of the new instance with their counterpart in the original
instance so that an assignment $\mu_0$ of clients in the original instance can be
translated naturally into an assignment $\mu_0'$ of clients in the new instance:
each client $c$ of the new instance is assigned to the same center than its counterpart
in the original instance.

%% Any facility of $L$ that is badly cut is forced to be part of the
%% solution. 
We say that a solution for the new instance, namely a set of centers $S$
and a mapping $\mu$, is \emph{valid} if it contains
all the badly cut facilities of $L$ and that all the clients
served by a badly cut facility $\ell$ in $L$ are mapped to $\ell$ in
$\mu$.

For a given instance of the $k$-clustering problem, the instance
that is defined as described above is a function of $\calD$, 
we thus refer to this instance by $\calI_{\calD,L}$. Note that the
set of candidate centers in $\calI_{\calD,L}$ is identical
to the set of candidate centers in the original instance. 

Therefore, for a given set of $k$ centers $S \subset \calA$ and an assignment
of clients to centers $\mu$, we denote
by $\cost(S, \mu)$ the cost of solution $(S,\mu)$ in the original instance
and by $\cost_{\calI_{\calD,L}}(S,\mu)$ the cost of solution $S$ in 
instance $\calI_{\calD,L}$. For any set of centers $S$, we let the optimal
assignment of clients to $S$ be $\mu_S$.

We let
\begin{align*}
  \nu_{\calI_{\calD,L}} =
  \max_{S \subseteq \calA \atop |S| \le k} \max(\cost(S,\mu_S)
  - (1+3\eps)\cost_{\calI_{\calD,L}}(S,\mu_S),\\
  (1-3\eps)\cost_{\calI_{\calD,L}}(S, \mu_S) - \cost(S, \mu_S)).
\end{align*}
This can be seen as how much a solution is ``distorted'' in the
instance $\calI_{\calD,L}$.
We say that an instance $\calI_{\calD,L}$ is \emph{good} if the following
conditions hold:
\begin{itemize}
\item $\nu_{\calI_{\calD,L}}  = O(\eps \cdot \cost(L, \mu_L))$
\item There exists a valid solution $\globalSh, \mu$ %% containing the
  %% facilities of $L$ that are badly cut and
  such that $\cost(\globalSh, \mu) \le
  (1+O(\eps))\cost(\opt,\mu_{\opt}) + O(\eps\cost(L,\mu_L))$.
\end{itemize}

The next lemma shows that an instance is good with probability at
least $1-\eps$.
We first step back and provide an
informal explanation on how to conclude from there; Assume that the instance is
good, then for the clients two things can happen: if the client
is not badly cut then it is at its location in the original
instance and none of its ring suffers a bad cut. In that case, its
assignment in the best valid solution is not distorted by the use of
portals. Otherwise, the client is badly cut, and then either the
facility $\ell$
it is assigned to in the local solution is also badly cut
or not. If $\ell$ is also badly cut then the client is assigned
to $\ell$ and so its service cost is 0 (no need of portals). Otherwise
its distance to the facility it is assigned to in the best valid solution is
not distorted by portals.
Then, since the cost of the best valid solution is small,
we can make use of dynamic programming and solve instance
$\calI_{\calD,L}$.

Here we conclude the section by showing that an instance is good
with probability $1-\eps$.
\begin{lemma}
  \label{lem:valid}
  Given a randomized split-tree decomposition $\calD$, the probability
  that $\calI_{\calD,L}$ is good is at least $1-\eps$.
\end{lemma}
\begin{proof}
  We first bound the probability that
  $\nu_{\calI_{\calD,L}}  \le O(\eps \cdot \cost(L, \mu_L))$.
  By definition, we have that
  for any solution $S, \mu$,   $  \cost(S,\mu) - \cost_{\calI_{\calD,L}}(S,\mu) \le$

  \begin{align*}
    &\sum_{\text{badly cut client } c} \dist(c,\mu(c))^p - \dist(\mu(c),\mu_L(c)))^p\\
    &\le  \sum_{\text{badly cut client } c}
    (1+3\eps)\dist(\mu(c),\mu_L(c))^p+ \\ &~~~~~~~~~~~~~~~~~~~~~~~~\frac{\dist(c,\mu_L(c))^p }{(\eps/(p+1))^p} -
    \dist(\mu(c),\mu_L(c))^p\\
    &\le \sum_{\text{badly cut client } c} 3\eps \cdot \dist(\mu(c),\mu_L(c))^p +
    \frac{\dist(c,\mu_L(c))^p }{(\eps/(p+1))^p}
  \end{align*}
  Where we have used Lemma~\ref{sometriangleineq} to go from the first to second line.
  Thus,
  \begin{equation}
    \label{eq:onedir}
    \cost(S,\mu) - (1+3\eps) \cost_{\calI_{\calD,L}}(S,\mu)
  \le  \sum_{\text{badly cut client } c}\frac{\dist(c,\mu_L(c))^p }{(\eps/(p+1))^p}
  \end{equation}

  Similarly, we have that $ \cost_{\calI_{\calD,L}}(S,\mu)- \cost(S,\mu) \le$

  \begin{align*}
    &\sum_{\text{badly cut client } c} \dist(\mu(c),\mu_L(c)))^p - \dist(c,\mu(c))^p\\
    &\le \sum_{\text{badly cut client } c} (1+3\eps) \dist(c,\mu(c))^p\\
    &~~~~~~~~~~~~~~~~~~~~~~~~~~+ \frac{\dist(c,\mu_L(c))^p}{(\eps/(p+1))^p} - \dist(c,\mu(c))^p\\
    &\le \sum_{\text{badly cut client } c} 3\eps \cdot \dist(c,\mu(c))^p +
    \frac{\dist(c,\mu_L(c))^p}{(\eps/(p+1))^p}    
  \end{align*}
  and so,
  \begin{align}
    \label{eq:otherdir}
    (1-3\eps)\cost_{\calI_{\calD,L}}(S,\mu)- \cost(S,\mu) \le \\
    \sum_{\text{badly cut client } c} \frac{\dist(c,\mu_L(c))^p}{(\eps/(p+1))^p}\nonumber
  \end{align}

  Now, observe that the right hand side of both Equations~\ref{eq:onedir}
  and~\ref{eq:otherdir} does not depend on $S$.
  Therefore, the expected value of $\nu_{\calI_{\calD,L}}$ is
  \begin{align*}
  E[\nu_{\calI_{\calD,L}}] \le \sum_{c} Pr[c \text{ badly cut}] \cdot
  \frac{\dist(c,\mu_L(c))^p}{(\eps/(p+1))^p} \le\\
  O(\eps^5 \cost(L)),
  \end{align*}
  where we have used Lemma~\ref{lem:badlycut}.
  We then apply Markov's inequality and obtain
  that $\calI_{\calD,L}$ satisfies the first condition with probability
  at least $1-\eps/3$. Let $\calE_{\nu}$ be the event that 
  $\calI_{\calD,L}$ satisfies the first condition.
  
  We now show that there exists a valid solution $\globalSh,\hat{\mu}$
  such that $\cost(\globalSh,\hat{\mu}) \le (1+O(\eps)) \cost(\opt,\mu_{\opt}) +
  O(\eps \cdot \cost(L,\mu_L))$.
  Consider an optimal solution $\opt$ and
  apply Proposition~\ref{prop:main} to $\opt$ and $L$.
  We let $\hL$, $\tL$, $\hF$, $\tF$ as defined by the
  proposition.
  
  We let $\calE_1$ be the event that there are at most
  $\eps^3 |\hL|$ facilities of $\hL$ that are badly cut.
  We have that by Lemma~\ref{lem:badlycut} the expected
  number of badly cut facilities in  $\hL$  is at most
  $\eps^5|\hL|$.
  Applying Markov's inequality we have 
  that $\calE_1$ holds with probability at least
  $1-\eps/3$.
  Condition on event $\calE_1$ happening.
  Consider $\globalSs$ as defined per Property~\ref{mainprop1} of the
  proposition. This solution contains $k - \Omega(\eps \cdot |\hL|)$.
  Thus, let $T$ be the temporary solution defined as $\globalSs$
  plus de badly cut facilities of $\hL$. Since we condition on
  event $\calE_1$ happening, we have that $T$ has at most
  $k$ centers. Hence, $T$ has at most $k$ centers with probability
  at least $1-\eps/3$.

  %% We now make use of Properties~\ref{mainprop1} and~\ref{mainprop2}
  %% to show that there exists a  solution $\globalSh$.
  
  We finally make use Property~\ref{mainprop2} of the proposition
  to incorporate the remaining badly cut facilities of $L$, \ie
  the badly cut facilites of $\tL$. We apply Property~\ref{mainprop2}
  to our random procedure for defining badly cut facilities in
  $\tL$. This shows that there exists a
  solution $\globalSh,\hat{\mu}$ such that 
  $\cost(\globalSh,\hat{\mu}) \le (1+O(\eps)) \cost(\opt,\mu_{\opt}) +  
  O(\eps \cdot \cost(L,\mu_L))$, with probability at least
  $1-\eps/3$.
  Taking a union bound over the probability that $\calE_{\nu}$
  and $\calE_1$ do not happen yields the lemma.
\end{proof}

\section{Structural Result}
\label{sec:struct}
We start by providing some intuition on what our main structural
aims to achieve.
Our goal is to show that there exists a near-optimal solution
which contains the badly cut facilities. To achieve this, consider
the bipartite graph obtained by having one vertex for each facility
of our current solution on one side, one vertex for each facility
of the optimal solution on the other side, and an edge from the vertex
corresponding to a facility of $\opt$ to the vertex corresponding to
the closest facility of our current solution.

Now, what we would like to argue is that if we pick a random facility 
$f$ of our current solution whose corresponding vertex in the bipartite
graph has at least one incomming edge from some vertex $v$,
then we can replace the facility corresponding to vertex $v$ in $\opt$
with $f$ and leave the cost unchanged, up to a factor $(1+1/k)$.
In the uncapacitated setting, one can formalize this intuition.
It then remains to address the issue of facilities that have no incomming
edge. Since the numbers of facilities of $\opt$ and our
current solution are the same, for each facility that has no incomming
edge, there is another facility with more than one incomming edge.
This allows to delete some facilities of $\opt$ to
make room for facilities of our current solution. See~\cite{abs-1812-08664} for
a complete proof of this argument in the uncapacitated setting.

For capacitated versions of the problem the picture changes drastically
since a replacement may incur a large reassignment of clients. In
some cases, reassigning the clients of the facility that was deleted to
the badly cut facility that replaces it may result in an arbitrarily bad
cost. Thus, finding a careful reassignment of the clients is crucial.

We now turn to the formal proof.
Let $\opt$ be an optimal solution and $\local$ be any solution.
%% For any $f \in \opt$, $\ell \in \local$, 
%% define $\reass(f, \ell)$ to be the cost of serving
%%  all the clients served by $f$ in $\opt$ by
%% $\ell$, namely $\sum_{c \in \Vglob(f)} \dist(c, \ell)$.
Define the \emph{charge} of a facility $f$ in a solution $S$
to be the total number of client assigned to $f$ in $S$.

\begin{proposition}
  \label{prop:main}
  Let $1/2 > \xi > 0$ be a fixed constant..
  Let $\opt$ be an optimal solution and $\local$ be any solution.
  Let $P$ be any random process such that each facility of $\opt$
  is selected with probability at most $\xi^2$.
  Then, there exists a partition of $\opt$ into two sets $\tF,\hF$
  and a partition
  of $\local$ into two sets $\tL,\hL$ such that
\begin{enumerate}
\item $|\hL| = |\hF|$, and so $|\tL| = |\tF|$.
\item\label{mainprop1} There exists $\hF^* \subseteq \hF$ of size
  at least
  $\xi |\hF|/3$ such that the set $\globalS^* = \opt - \hF^*$ is
  a solution of cost at most $\cost(\opt) +
  O(\xi (\cost(\opt) + \cost(\local)))$.
\item\label{mainprop2}
  Let $\tF^*$ be the set of facilities selected by the random process
  $P$.
  Then, there exists a 1-to-1 mapping $\phi: \tF \mapsto \tL$
  %  and an integer
  that satisfies the following.
  Let $\tL^* = \bigcup_{f \in \tF^*} \phi(f)$.
  With probability at least $1-\xi$,
  the solution $\globalS = \globalS^* - \tF^* \cup \tL^*$ and where
  each client served by a facility $\ell \in \tL^*$ in solution $L$
  is served by $\ell$ in solution $\globalS$,
  has cost at most   $$\cost(\opt) +   O(\xi (\cost(\opt) + \cost(L)))$$

\end{enumerate}
\end{proposition}

We consider the following bipartite graph $\Phi = (A,B, E)$ with
both capacities and costs (or weights) on the edges defined as follows.
The set $A$ contains one vertex for each facility of $\opt$ plus a special
vertex $t$.
The set $B$ contains one vertex for each facility of $\local$ plus a special
vertex $s$. We slightly abuse
notation and call the vertex representing facility $f$ by $f$ as well.
The set of edges is as follows: for each facility $f \in \opt$ and $\ell \in \local$,
for each client $c$ that is served by $f$ in $\opt$ and $\ell$ in $\local$, add a
directed edge $e$
from $f$ to $\ell$ in $\Phi$. We refer to $e$ as the edge corresponding to client $c$.
The capacity of the edge is $2$ and the
cost of the edge is $\dist(c, f) + \dist(c,\ell) = g_c + \ell_c$. Note that this may create
parallel edges -- parallel edges are kept in $\Phi$. 

Furthermore, for  each vertex of $f \in A - \{t\}$, 
we add $\lfloor \capac/2 \rfloor $ directed edges from $s$ to $f$ each with capacity 2 and cost 0.
% Each such 
% edge has capacity $\min(\capac, 2\lambda_f)$ and cost 0.
Finally, from each vertex
of $\ell \in B - \{s\}$, we add $\lfloor \capac/2 \rfloor$ edges directed from $\ell$ to $t$
each with capacity 2 and cost 0.
%% Recall that for the moment we assume $\capac$ is a multiple of 2 so the total
%% maximum flow from $s$ to $f \in A - \{t\}$ or from $\ell \in B- \{s\}$ to $t$
%% is at most $\capac$.

% \note{can remove the min cause the outdegree of each $f \in B - s$ is at most $5 \lambda_f$ anyway.}

\paragraph{Preprocessing step when $\eta$ is not a multiple of 2}
We now apply a preprocessing step for the case when $\capac$ is not a
multiple of 2. We assign a fractional weight of $1/\capac$ to each edge
that connect a vertex of $A-\{t\}$ to a vertex of $B - \{s\}$ of
$\Phi$. This defines a fractional matching over the vertices of $A-\{t\} \cup B - \{s\}$
where each vertex of $A-\{t\}\cup B - \{s\}$
that serves $\capac$ clients is such that the total weight of the edges adjacent to
it is 1. Therefore, there exists an integral matching where each vertex of f $A-\{t\} \cup B - \{s\}$
that serves $\capac$ clients is matched (see \eg~\cite{matchingtheory}). Consider
such a matching and delete the edge of the matching. We refer to the clients
corresponding to the deleted edge by the \emph{deleted} clients.
The degree of each vertex after the preprocessing step differs by at most 1 from
the original degree.

\bigskip

In the remaining, we let $\capac' = \capac-1$ if $\capac$ is not a multiple of 2, and
$\capac' = \capac$ otherwise. Hence $\capac'$ is a multiple of 2.
For each facility $f \in A - \{t\}$, 
we denote by $\capac(f)$ the outgoing degree of $f$ after the preprocessing step and
for each $f \in B - \{s\}$, we denote by $\capac(f)$ the incoming degree of $f$ after
the preprocessing step.
We put a demand of $2\lfloor \capac(f)/2 \rfloor$ on each vertex of $f \in A-\{t\}$ and a demand
of $2\lfloor \capac(f)/2 \rfloor$ on each vertex of $f \in B-\{s\}$.

%% We assume $\capac = 2y + 1$ for some integer $y$.
%% To do so, consider all the facilities of $A$ and all the facilities of $B$
%% that serve $\capac$ clients. There exists a matching between facilities
%%     of $A$ and facilities of $B$ where each full facility of $A$ and $B$ is matched.
%%     We now proceed as described in the above process to route the flow, except that
%%     between each pair of matched facilities we decrease the total amount of flow traveling
%%     from one to the other by 1. This is always possible since there was an edge between them
%%     and so, at least one client served by both.
%%     This ensures that each facility of $B-\{s\}$ receives a flow of at most
%%     $2 \lfloor \capac/2 \rfloor$ and each facility $f$ in $A-\{t\}$ receives a flow
%%     of at least $2 \lfloor \capac(f)/2 \rfloor$.

% We also define an edge from $\ell$ to $f$ of cost
% $\dist(c, f) + \dist(c,\ell)$ and capacity 6.
% We allow these edges to be bidirectional as we allow the flow
% to go from both $f$ to $\ell$ and $\ell$ to $f$, the capacity in each direction
% is 3, the cost is the same.

\begin{lemma}
  \label{Lemma:maxflow}
  There exists a flow $\calF_0$ in $\Phi$ from $s$ to $t$ of cost
  at most $\cost(\local) + \cost(\opt)$ and that satisfies:
  \begin{itemize}
  \item \textbf{Integrality:} each edge between $A-\{t\}$ and $B-\{s\}$ receives a flow of either 0 or 2.
  \item \textbf{Demand:} each vertex $f \in A-\{t\}$ receives a flow of at least $2 \lfloor \capac(f)/2 \rfloor$;
    each vertex $f \in B-\{s\}$ receives a flow of at least $2 \lfloor \capac(f)/2 \rfloor$.
  \end{itemize}
\end{lemma}
\begin{proof}
  We will show the following claim:
  \begin{claim}
    \label{claim:preproc}
    There exists a flow that satisfies the demand constraint
    and the capacities
    of the edges, but that does not necessarily
    satisfy the integrality constraint.
  \end{claim}
  Then, assuming Claim~\ref{claim:preproc} the lemma follows:
  %given edge capacities, weights, and demands that are all equal to $0 \mod 2$,
  Classic results (\eg \cite{matchingtheory}) on the integrality of flows
  show that
  if the edges all have capacities 2,
  the demands are multiple of 2, and there exists a fractional
  flow satisfying the demands and capacities, 
  %% then there exists a flow that satisfies
  %% the demand and capacity constraints,
  then there is 
  a flow that sends either a flow of 0 or 
  a flow of 2 in each edge and that satisfies the demand constraints.
  Thus, we turn to the proof of Claim~\ref{claim:preproc}

  %% \begin{Proof}[Proof of Claim~\ref{claim:preproc}]
    
  Consider sending a flow from $s$ to each vertex $f$ of $A - \{t\}$
  of a value $2 \lfloor \capac(f)/2 \rfloor$.
  Since the total capacity from $s$ to $f$ is $2 \lfloor \capac(f)/2 \rfloor$
  %are all at most $\capac$ % \min(\capac, 2\lambda_f)$
  this is possible and the current
  cost of the flow is 0.
  Now, consider for each non-deleted client $c$ served by $f$ in $\opt$, to
  send a flow of $1$ from $f$ 
  to the facility of $\local$ that serves it in solution $\local$.
  This corresponds to sending a flow of 1 through the edge corresponding to client $c$.
  By the definition of the graph, for each such client $c$ there exists
  an edge with capacity 2 between the two facilities. This ensures that the
  demand at each facility $f \in A-\{t\}$ is met.
  %and the cost of this
  %edge corresponds to the cost paid by client $c$ in both solutions.
    
  Finally, observe that each vertex $f$ of $B - \{s\}$ receives a flow
  that corresponds to the number of non-deleted clients served by the center in
  solution $\local$. % and so at most $\capac$.
  Thus, it is possible to complete the assignment by sending the flow arriving
  in each vertex $f$ of $B - \{s\}$ to $t$ using the edge of capacity
  $2 \lfloor \capac(f)/2 \rfloor$ and cost 0 and the demand at $f$ is met. %The bound on the cost of the flow follows.
%  \end{Proof}
%% \note{This need more details and the correct citation.}
\end{proof}

%% Observe  that
%% each vertex $f \in A - \{t\}$ receives a flow of at least $\capac(f) - 1$.
Let $\calF$ denote a maximal integral flow satisfying the
demand and integrality constraints, as per Lemma~\ref{Lemma:maxflow}.
%% Note that there exists a flow satisfying the demand constraints since $\calF_0$
%% is such a flow so $\calF$ carries a flow that has value at least the value of
%% the flow $\calF_0$.
We say that a facility is \emph{saturated} if the total flow it receives is $\eta'$.
% $2\lfloor \capac/2 \rfloor$.
We say that an edge is $\calF$-\emph{saturated} if the total flow
in the flow $\calF$ that goes through the edge is 2.

\begin{lemma}
  The cost of $\calF$ is at most $2(\cost(\local) + \cost(\opt))$.
\end{lemma}
\begin{proof}
  This follows from the fact that when all the edges of the graph are saturated
  the total cost is $2(\cost(\local) + \cost(\opt))$.
\end{proof}

% Since the sum of the capacities of the edges
% towards $\ell \in B$ is at most $5\capac$, we have
% \begin{fact}
%   For any $\ell \in B$, $f \in A$,
%   the total flow from $\ell$ to $f$ is at most
%   $3 \capac$.
% \end{fact}
\newcommand{\boundonsize}{\eta'}
We now define $U$ to be the set of facilities of $A$ such that
$\eta(f) \ge \boundonsize/2$,
%is at least  $\eta'/2$, %$\lfloor \capac/2 \rfloor +1$,
\ie $U = \{f \mid f \in A, \capac(f) \ge \boundonsize/2 \}$. We will refer to
the facilites of $U$ as \emph{heavy} facilities.
Let $\Lambda$ be the set of facilities of $U$ whose corresponding
vertices in $\Phi$ that are saturated by  %have an incoming flow of $\capac$
flow $\calF$. Define $\bar{\Lambda} = A - \{t\} - \Lambda$.
Let $\zeta$ be the set of facilities of $B - \{s\}$ that are saturated
by flow $\calF$. Define $\bar{\zeta} = B- \{s\} - \zeta$.

We now aim at matching vertices of $\Lambda$ and $\zeta$ to vertices
of $B - \{s\}$ and $A - \{t\}$ respectively.
We will make use of the following classic theorem (see \eg~\cite{matchingtheory}).
\begin{theorem}[\cite{matchingtheory}]
  \label{thm:matchingtheory}
  Let $G = (A,B,E)$ be a bipartite graph with edge weights
  $w : E \mapsto \R_+$.
  Let $M_0 : E \mapsto [0,1]$ be a fractional matching
  of weight $W = \sum_e M_0(e) \cdot w(e)$. 
  There exists an integral matching $M_1 : E \mapsto \{0,1\}$
  that satisfies:
  \begin{itemize}
  \item Each vertex $u \in A \cup B$ such that $\sum_{(u,v) \in E} M_0((u,v)) =1$ is matched in $M_1$, \ie 
    $\exists (u,v) \in E$ s.t. $M_1((u,v)) = 1$, and;
  \item The weight of $M_1$ is at most $W$, \ie $\sum_{e \in E } M_1(e) \cdot w(e) \le W$, and;
  \item $M_1((u,v)) = 1 \implies M_0((u,v)) \neq 0$.
  \end{itemize}
\end{theorem}

Consider rescaling the amount of flow $\calF$ sent through each edge by a factor
%% $1/(2 \lfloor \capac/2 \rfloor)$
$1/\boundonsize$
and denote by $\calM_0$ the underlying flow. Seeing $\calM_0$ as a matching 
of weight at most $2(\opt + \cost(\local))/\boundonsize$ %(2 \lfloor \capac/2 \rfloor)$,
we have the following application of Theorem~\ref{thm:matchingtheory}:
\begin{cor}
  \label{cor:matchingcost}
  There exists an integral matching $\calM$ in $\Phi$ that satisfies:
  \begin{enumerate}
  \item Each facility of $\Lambda$ is matched to a facility
    of $B - \{s\}$, and;
  \item Each facility of $\zeta$ is matched to a facility
    of $A - \{t\}$, and;
  \item The weight of the matching is at most $2(\cost(\opt) + \cost(\local))/\boundonsize$, and; %(2 \lfloor \capac/2 \rfloor)$, and;
  \item \label{cor:cst:pos} If a facility $f$ is matched to a facility $\ell$, then it must
    be that the flow $\calF$ going from $f$ to $\ell$ is positive.
  \end{enumerate}
\end{cor}

We define $\calM_A$ to be the set of vertices of $A - \{t\}$ that
are matched and $\calM_B$ the set of vertices of $B - \{s\}$ that are 
matched. Note that $\Lambda \subseteq \calM_A$ and $\zeta \subseteq \calM_B$.
% -- note that it could be that $\Lambda \neq \calM_A$ or  $\zeta \neq \calM_B$.

Consider a facility $\ell \in \calM_B$, we define the following mapping. 
Let $f(\ell)$
be the facility of $A$ that is matched to $\ell$ in $\calM$.

We now consider each pair of matched vertices $\ell, f(\ell)$ and define a function $p$ that
maps each edge incoming to $\ell$ to either an edge outgoing of $f(\ell)$ or to $f(\ell)$
directly. 
For each vertex $\ell \in \calM_B$, define $t(\ell)$ and $\bar{t}(\ell)$ to be respectively the numbers of
non-$\calF$-saturated and $\calF$-saturated edges ingoing to $\ell$ and not originating
from $f(\ell)$. For each vertex
$f \in \calM_A$ define $s(f)$ and $\bar{s}(f)$ to be respectively the numbers of non-$\calF$-saturated and $\calF$-saturated
edges outgoing from $f$ and not going to $\ell$. 

The mapping $p$ is defined as follows. Consider a pair of matched vertices $\ell,f(\ell)$.
\begin{enumerate}
\item If $t(\ell) > s(f(\ell))$, choose an arbitrary subset of size $s(f(\ell))$ among the non-$\calF$-saturated
  edges incoming to $\ell$ and define a one-to-one mapping from these edges to the edges in $s(f(\ell))$.
  For the $t(\ell) - s(f(\ell))$ remaining edges, map them to $f(\ell)$. This defines the mapping $p$ for the 
  non-$\calF$-saturated  edges incoming to $\ell$ for the case ($t(\ell) > s(f(\ell))$).
  
  Otherwise, when $t(\ell) \le s(f(\ell))$, simply define an arbitrary injective function from the non-$\calF$-saturated
  edges incoming to $\ell$ to the non-$\calF$-saturated edges outgoing from $f(\ell)$.
  %% This defines the mapping $p$ for the 
  %% non-$\calF$-saturated  edges incoming to $\ell$ otherwise.
\item Proceed similarly with the $\calF$-saturated edges that are incoming to $\ell$ and outgoing from $f(\ell)$.
\end{enumerate}

We now consider vertices of $A - \calM_A-\{t\}$. For each such vertex $f$, for each edge $e$ outgoing
from $f$, we define $S(e)$ to be the sequence obtained by recursively applying the mapping $p$
(\ie $S(e) = e, p(e), p(p(e)), \ldots,$) until we can't apply $p$ again, namely
either we reach an edge $e'$ such that $p(e') = f'$ where $f' \in \calM_A$, or we reach an edge $e' = (f',\ell')$
where $\ell' \in B - \calM_B - \{s\}$.
Let $\calS$ be the set of all the sequences defined above.
We have the following lemma.
\begin{lemma}
  \label{lem:noinfinite}
  For each sequence $S(e) \in \calS$, we have that
  each edge of the graph appears at most once in $S(e)$ and so, $S(e)$ is finite.
  Moreover, for each edge $e'$ of the graph, there is at most one sequence in $\calS$  containing $e'$.
  % The sequence $S(e)$ is finite.  
\end{lemma}
\begin{proof}
  % To prove the lemma, we assume toward contradiction that $e$ is not-$\calF$-saturated consider the path defined by
  % $e$ and the sequence $S(e) = p(e)$, $p(p(e))$, $p(p(p(e)))$, $\ldots$,
  % $p^j(e)$, where $p^j(e) = (f_j,\ell_j)$ is the first edge of this sequence such that $\ell_j$ is unmatched.
  We first argue that $S(e)$ is finite. 
  Let $e = (f,\ell)$ with $f \in A - \calM_A - \{t\}$.
  Recall that $p$ maps the edges adjacent to a facility $\ell \in \calM_B$ and not coming from $f(\ell)$ either 
  to a facility $f(\ell)$ in which case, the sequence stops, or
  is an injective mapping to the edges outgoing from $f(\ell)$ and not going to $\ell$.

  % for all facilities $\ell \in B - \{s\}$, to the non-$\calF$-saturated edges adjacent to $f(\ell)$.
  Furthermore, observe that for any edge $(f^j,\ell^j)$ of the sequence, $p((f^j,\ell^j))$ is an edge
  which starts at a matched vertex. Therefore, except for the first edge $e$, no edge of the sequence
  is adajcent to $f$ since $f$ is unmatched, \ie no edge in the sequence $p(e),p(p(e)),\ldots$ is adjacent to $f$.

  % Therefore, since $e$ originates from a facility that is not matched, an edge cannot appear twice in the sequence
  % $S(e)$. More concretely, 
  Assume towards contradiction that there is an edge that appears twice in the sequence
  and consider the first one in the order of the sequence. 
  Let $(v_i,u_i)$ be this edge. By the above argument, we have that $v_i \neq f$ since otherwise there
  would be an edge in the subsequence $p(e),p(p(e)),\ldots$ that is adjacent to $f$.

  Thus, we have $v_i \neq f$, and so $p^{-1}((v_i,u_i))$ is also twice in the sequence since $p$ is injective on the edges.
  This is a contradiction since $(v_i,u_i)$ is the first
  one of the sequence, it follows that $S(e)$ is finite.
  
  Finally, since for any $S((f,\ell)) \in \calF$, we have that $f$ is an unmatched vertex, the edge $(f,\ell)$
  cannot appear in another sequence $S(e) \in \calF$. Thus applying the same reasonning as above, an edge cannot
  appear in two different sequences.
\end{proof}

We now distinguish two types of sequences. We say that an edge $e$ is a \emph{route-to-matched} if the
sequence $S(e)$ stops at a vertex $f \in \calM_A$, and a \emph{route-to-unmatched} if the sequence
$S(e)$ stops at a vertex $\ell \in B - \calM_B - \{s\}$.
We have the following lemma.
\begin{lemma}
  \label{cl:onlysaturated}
  Consider a facility $f \in A-\calM_A-\{t\}$ and such that
  $\eta(f) \ge \boundonsize/2$.
  %\lfloor \eta/2 \rfloor +1 $.
  The number of edges $e$ adjacent to $f$
  and such that $S(e)$ is a route-to-unmatched sequence
  is at most $\boundonsize/2 -1$ %\lfloor \eta/2 \rfloor $.
  % Any edge $e=(f,\ell)$
  % in $\Phi_2$ corresponds to a $\calF$-saturated edge of $\Phi$.
\end{lemma}
\begin{proof}
  Since $f$ is unmatched,
  we have that the total flow going through $f$ in the flow $\calF$ is at most $\boundonsize$.
  %$2\lfloor \eta/2\rfloor$.
  Thus, there are at most $\boundonsize/2-1$  edges that adjacent
  to $f$ and that are $\calF$-saturated.
  We will show that for each edge $e$ adjancent to $f$ such that $S(e)$ is a route-to-unmatched
  sequence, we have that $e$ is $\calF$-saturated.
  
  % Recall that the total flow
  % of $\calF$ going through $f$ is less than $\capac$ since otherwise, $f$ would be matched.
  % On the other hand, $f$ could receive from $s$ a total flow of $\capac$ and since $\eta(f)\ge\eta/3$,
  % the total capacity of its outgoing edges in $\Phi$ is
  % at least $\capac$.
  % Let $e = (f,\ell)$ and $e \in \Phi \cap \Phi_1$.
  
  Now, suppose towards contradiction that there exists an edge $e$ such that $S(e)$ is a route-to-unmatched sequence
  that is not-$\calF$-saturated and consider the
  path induced by the sequence $S(e)$. By Lemma~\ref{lem:noinfinite}, this sequence is finite and so
  let $(f_j,\ell_j)$ be the last edge of the sequence. Since $S(e)$ is a route-to-unmatched, 
  $\ell_j$ is unmatched.

  Since we have that $e$ is not-$\calF$-saturated and by definition of $p$, a direct induction
  shows that all the edges in the sequence $S(e)$ are not-$\calF$-saturated, and so these are all
  edges with positive capacities in the residual graph $\Phi^{\calF}$. Moreover, observe that
  for each matched pair $\ell_x,f(\ell_x)$, Corollary~\ref{cor:matchingcost}, Property~\ref{cor:cst:pos}
  implies that there are at least 2 units of flow going from $f(\ell_x)$ to $\ell_x$ in flow $\calF$. This induces
  an edge with positive capacity from $\ell_x$ to $f(\ell_x)$ in $\Phi^{\calF}$.

  Thus, consider the subgraph of $\Phi^{\calF}$ induced by the edges of $S(e)$ and the edges between matched
  pairs $\ell_x,f(\ell_x)$ and consider a simple path from $f$ to $\ell_j$ in this graph. This path uses each
  edge at most once and so it is possible to route at least one unit of flow through this path without violating
  the capacities of the edges of the path.

  Furthermore, $\ell_j$ and $f$ are not matched
  and so there is at least one edge with positive capacity from $\ell_j$ to $t$ and an edge with positive capacity
  from $s$ to $f$ in $\Phi^{\calF}$. Therefore there is a path with positive capacity from $s$ to $t$ in $\Phi^{\calF}$.
  Furthermore, observe that routing a unit of flow through this path can only increase the flow going through
  any of the vertices of the graph. Thus, %we reach a contradiction to the maximality of $\calF$ under the demand constraints:
  there is a flow with higher value which satisfies the demand constraints, a contradiction to the maximality of $\calF$
  that concludes the proof.
\end{proof}

\paragraph{Assignment $\mu$}
For each facility $f \in U - \calM_A$, namely an unmatched heavy facility, consider the edges $e=(f,\ell)$
such that
$S(e)$ is a route-to-matched sequence. For the client $c$ associated with edge $e$, we let $\mu(c)$ map to the matched
vertex of $\calM_A$ that terminates the sequence $S(e)$. Let $\calC_1$ be the set of these clients.
%% We define the \emph{path associated with sequence $S(e)$} to be the path made of the following edges:
%% $e=(f,\ell), \dist(\ell,f(\ell)), p(e)=(f(\ell), \ell_1), \dist(\ell_1,f(\ell_1)),\ldots$.
For clients in $\calC- \calC_1$ (including deleted clients), we les $\mu(c)$ be the facility that
serves it in $\opt$.

\paragraph{Sequences to paths}
For each facility $f \in U - \calM_A$, namely an unmatched heavy facility, consider the edges $e=(f,\ell)$.
For each such edge $e$, we define the \emph{path associated to sequence $S(e)$} as follows. The first edge
of the path is $(f,\ell)$, the second edge of the path is $(\ell,f(\ell))$, the third edge of the path
is $p(e) = (f(\ell), \ell_1)$. For $i>1$, the $2i$-th and $2i+1$-st edges of the path are
edges $(f(\ell_{i-1}),\ell_i)$ and $(\ell_i,f(\ell_i))$. If there are multiple edges $(\ell_i,f(\ell_i))$, the one
with smallest weight is chosen. We let $\calP(e)$ denote the path associated to edge $S(e)$.
By the triangle inequality, the length of the path is simply the sum of the weights of the edges. 

We show the following lemma, which will be used in two different ways:
\begin{enumerate}
\item to bound the cost of reassigning a client whose corresponding edge is a route-to-matched client
  to the facility of $\calM_A$ at the end of the sequence;
\item to bound the cost of reassigning a client whose corresponding edge is a route-to-unmatched client
  to the facility of $B - \calM_B$ at the end of the sequence.
\end{enumerate}
\begin{lemma}
  \label{lem:path}
  The sum over all facility $f \in U - \calM_A$, of the sum over all edges $e=(f,\ell)$ of the length
  of the paths associated to $S(e)$ is at most $4(\cost(\opt) + \cost(L))$.
  In other words,
  $$
  \sum_{f \in U - \calM_A} \sum_{e = (f,\ell)} \text{length}(\calP(e)) \le 4(\cost(\opt) + \cost(L)).
  $$
\end{lemma}
\begin{proof}
  %% Now, we want to bound the sum of the length of the edges in all the paths associated to the route-to-matched sequences.
  Observe that the by Lemma~\ref{lem:noinfinite}, the paths are edge disjoint, except for the edges of the path
  that are connecting two vertices that are matched together (namely, the even edges of the path).
  More concretely, in the path $ e=(f,\ell), \dist(\ell,f(\ell)), p(e) = (f(\ell),\ell_1), \dist(\ell_1,f(\ell_1)),\ldots $ associated to
  sequence $S(e)$, the edges $(f,\ell)$, $p((f,\ell))$, $p(p((f,\ell)))$, $\ldots$ appear in at most one sequence.
  Thus, the sum over all sequences of the edges that are not connecting two matched vertices is bounded by the total
  sum of edge weights of the graph and so at most $(\cost(\opt)+ \cost(L))$.

  We now bound the number of times $\dist(\ell,f(\ell))$ is going to appear in the sum of the lengths of the paths of all the sequences.
  We first observe that the number
  of paths in which this edge appears is bounded by the incoming degree of $\ell$ which corresponds to the number of clients
  served by $\ell$ in $L$ and so at most $\boundonsize$. Thus, we have that $\dist(\ell,f(\ell))$ appears at most $\boundonsize$ times
  in the sum.
  Finally, Corollary~\ref{cor:matchingcost}, Property~\ref{cor:cst:pos}, combined with the triangle inequality shows that
  $\sum_{\ell \in \calM_B} \dist(\ell,f(\ell)) \le \sum_{\ell \in \calM_B} w(e((\ell,f(\ell))) \le
  2(\cost(\opt)+ \cost(L))/\boundonsize$, where $e((\ell,f(\ell))$ is the edge matching $\ell$ to $f(\ell)$
  and $w$ its weight. Thus, since $\dist(\ell,f(\ell))$ appears at most $\boundonsize$ times for each matched pair $\ell,f(\ell)$,
  we have that the total cost induced by these edges is at most $2(\cost(\opt)+\cost(L))$.
\end{proof}

\paragraph{Remark on the case $p>1$.}
For any objective where the cost of assigning client $c$ to facility $f$ is $\dist(c,f)^p$, for $p>1$,
Lemma~\ref{lem:path} does not allow to relate the cost of assigning a client $c$ to the facility that
is at the end of path of the sequence $S(e)$ where $e$ is the edge corresponding to client $c$.
Indeed, for example for $p=2$, the cost for a client in $\calC_1$ is going to be the square of the
sum of the weights of the edges in the path and this cannot be related to the cost of an optimal solution and
the cost of the local solution
directly since the solutions pays the sum of the lengths squared (instead of the square of the sum of the lengths).

The way to handle this is to modify the definition of the length of a path associated to $S(e)$.
Consider first a route-to-matched sequence
$S(e)= e, p(e), p(p(e)), \ldots$ and let $e=(f,\ell), p(e) = (f(\ell),\ell_1), p^i(e) = (f(\ell_{i-1}),\ell_i)$, for $i>1$.
Let $f(\ell_{s-1})$ be the matched vertex that terminates the sequence.
We define the length of the path associated to $S(e)$ as
$\text{length}(\calP(e)) = $
\begin{align*}
  (\dist(f,\ell)+\dist(\ell,f(\ell))^p +\\
  \sum_i (\dist(f(\ell_{i-1}),\ell_i) + \dist(\ell_i,f(\ell_i)))^p.
\end{align*}

Observe that $(a + b)^p \le 2^{p}(a^p + b^p)$ and so we have that
\begin{align*}
  \text{length}(\calP(e)) \le
  2^p(\dist(f,\ell)^p+\dist(\ell,f(\ell))^p + \\
  \sum_i (\dist(f(\ell_{i-1}),\ell_i)^p + \dist(\ell_i,f(\ell_i))^p).
\end{align*}
Now recall that for each edge of the graph $(f',\ell')$, the weight is given by $\dist(c,f')^p + \dist(c,\ell')^p$ where $c$ is the client
associated to the edge. Thus we have that $\dist(f',\ell')^p \le 2^p(\dist(c,f')^p + \dist(c,\ell')^p)$.
Therefore, mimicking the proof of Lemma~\ref{lem:path}, we have that
$\sum_{f \in U - \calM_A} \sum_{e = (f,\ell)} \text{length}(\calP(e))$
is at most $2^{O(p)} (\cost(\opt) + \cost(L)).$

Now, the length of a path $\calP(e)$ does not correspond to the cost
of assigning of the client $c$ associated to edge $e$ to
the facility at the end of the path. Instead, it corresponds to
the cost of assigning $c$ to $f(\ell)$ plus the cost of assigning
the client $c_1$ of $f(\ell)$ to $f(\ell_1)$ whose edge is in
the sequence, plus the cost of assigning the client $c_2$ of
$f(\ell_1)$ to $f(\ell_2)$ whose edge is in the sequence,
and more generally the cost of assigning
the client of $f(\ell_i)$ whose edge is in the sequence
to $f(\ell_{i+1})$, for all $i < s$.
%% if the capacity of $f(\ell) > \eta$,
%% then to the cost of assigning in addition an arbitrary client of $f(\ell)$
%% to $f(\ell_1)$, and more generally, for all $i$, if the capacity of $f(\ell_i)$ becomes greater than $\eta$, we assign
%% an arbitrary client of $f(\ell_i)$ to $f(\ell_{i+1})$ for all $i<s$.
Indeed, the total cost of such a reassignment is given by 
$(\dist(f,\ell)+\dist(\ell,f(\ell))^p + \sum_i (\dist(f(\ell_{i-1}),\ell_i) + \dist(\ell_i,f(\ell_i)))^p = \text{length}(\calP(e))$
and so bounded by $2^{O(p)}(\cost(\opt) + \cost(L))$.

In the case of $p>1$, we let $\mu^p$ be the reassignment defined above. It is easy to see that
if $\mu$ meets the capacities then $\mu^p$ also meets the capacities.

%% Instead of assigning $c$ to $f(\ell_{s-1})$ we assign $c$ to $f(\ell)$ and, if the capacity of $f(\ell) > \eta$,
%% then assign an arbitrary client of $f(\ell)$
%% to $f(\ell_1)$, and more generally, for all $i$, if the capacity of $f(\ell_i)$ becomes greater than $\eta$, we assign
%% an arbitrary client of $f(\ell_i)$ to $f(\ell_{i+1})$ for all $i<s$. Note that this yields an assignment that satisfies
%% all capacities since $\mu$ satisfies all capacities.

%% which can now be related to the $\cost(\opt) + \cost(L)$.
%% The bound on Property~\ref{lem:intersol:cond3} of Lemma~\ref{lem:intersol} now becomes
%% $O(2^p (\cost(\opt) + \cost(L)))$.

%% \note{Here we need to update the fact that each heavy facility may actually have one extra guy that was not counted due
%% to the rounding to the closest multiple of 3.}

We now turn to show that the capacities are met for assignment $\mu$. 
\begin{lemma}
  \label{lem:intersol}
  Consider the solution defined by the set of centers of $\opt$ together with the assignment $\mu$.
  This solution satisfies the following properties:
  \begin{enumerate}
  \item\label{lem:intersol:cond1} Each facility $f$ of $\opt$ whose corresponding vertex is unmatched, \ie $f \not\in \calM_A$,
    is assigned at most $\lfloor \eta/2 \rfloor$ clients. In other words, $|\{c \mid \mu(c) = f\}| \le \lfloor \eta/2 \rfloor$.
  \item\label{lem:intersol:cond2} Each facility $f$ of $\opt$ whose corresponding vertex is matched, \ie $f \in \calM_A$,
    is assigned at most $\eta$ clients. In other words, $|\{c \mid \mu(c) = f\}| \le \eta$.
  \item\label{lem:intersol:cond3} For each client $c \in \calC-\calC_1$, its cost is identical to its cost in $\opt$.
    Moreover, $\sum_{c \in \calC_1} \dist(c,\mu(c)) \le 2^{O(p)}(\cost(\opt) + \cost(L))$.
  \end{enumerate}
  
\end{lemma}
\begin{proof}
  We first prove Property~\ref{lem:intersol:cond1}. From Lemma~\ref{cl:onlysaturated}, 
  the only clients that
  are assigned to an unmatched facility in $\mu$ are the one for which sequence $S(e)$ of the corresponding edge $e$
  is a route-to-unmatched plus possibly one deleted client. It follows that the total number of clients assigned
  is $\boundonsize/2 = \lfloor \eta/2 \rfloor$. 

  To prove Property~\ref{lem:intersol:cond2}, we start with the following observation. Consider a pair of matched vertices,
  $\ell,f(\ell)$. The total number of new elements that can be assigned to $f(\ell)$ in mapping $\mu$ is, by definition of $p$,
  the number of edges
  that are incoming to $\ell$ and not originating from $f(\ell)$ minus the number of edges outgoing from $f(\ell)$ and
  not going to $\ell$, or in other words $\max(\bar{t}(\ell)- \bar{s}(f(\ell)), 0) + \max(t(\ell)  - s(f(\ell)),0)$.
  Let $m_{\ell}$  and $\bar{m}_{\ell}$ respectively denote the number of non-$\calF$-saturated and $\calF$-saturated
  edges between $\ell$ and $f(\ell)$

  It follows that the total number of non-deleted clients served by $f(\ell)$ in assignment $\mu$ is at most
  \begin{align*}
    \nu =\\
    \bar{s}(f(\ell)) + s(f(\ell)) + m_{\ell} + \bar{m}_{\ell} + \\
    \max(\bar{t}(\ell)- \bar{s}(f(\ell)), 0) + \max(t(\ell)  - s(f(\ell)),0).
  \end{align*}
  We aim at showing that $\nu$ is at most $\boundonsize$.
  We have the following equations:
  \begin{itemize}
  \item $ \bar{s}(f(\ell)) + s(f(\ell)) + m_{\ell} + \bar{m}_{\ell} \le \boundonsize$, since the degree of $f(\ell)$ in $\Phi$
    (after preprocessing) is at most $\boundonsize$;
  \item $ \bar{t}(\ell) + t(\ell) + m_{\ell} + \bar{m}_{\ell} \le \boundonsize$, since the degree of $\ell$ in $\Phi$
    (after preprocessing) is at most $\boundonsize$;    
  \item $2\lfloor \frac{\bar{t}(\ell) + t(\ell) + m_{\ell} + \bar{m}_{\ell}}{2} \rfloor  \le 2 \bar{t}(\ell)  + 2 \bar{m}_{\ell} \le
    \boundonsize$
    % 2\lfloor \eta/2\rfloor $,
    since the flow $\calF$ going through $\ell$ is at least $\lfloor \frac{\bar{t}(\ell) + t(\ell) + m_{\ell} + \bar{m}_{\ell}}{2} \rfloor$
    by the definition of the demand at $\ell$ and at most $\boundonsize$ since the outgoing capacity from $\ell$
    is $\boundonsize$. Moreover, each edge of $\bar{t}(\ell)$ and $\bar{m}_{\ell}$ carries 2 units of flow.
  \item $2\lfloor \frac{\bar{s}(f(\ell)) + s(f(\ell)) + m_{\ell} + \bar{m}_{\ell}}{2} \rfloor  \le 2 \bar{s}(f(\ell))  + 2 \bar{m}_{\ell} \le
    \boundonsize$,
    %2\lfloor \eta/2\rfloor $,
    for the same reason than the above case.
  \end{itemize}

  In the case where either both $\bar{t}(\ell) \ge \bar{s}(f(\ell))$ and $t(\ell) \ge s(f(\ell))$ or both 
  $\bar{t}(\ell) \le \bar{s}(f(\ell))$ and $t(\ell) \le s(f(\ell))$, then we have that $\nu$ is at
  most $\boundonsize$ by combining directly with the first two equations of the above list.

  We thus turn to the case where $\bar{t}(\ell) \ge \bar{s}(f(\ell))$ and $t(\ell) \le s(f(\ell))$.
  First, if $\bar{t}(\ell) = \bar{s}(f(\ell))$, then both max are 0 and so 
  $\nu \le  \bar{s}(f(\ell)) + s(f(\ell)) + m_{\ell} + \bar{m}_{\ell} \le \boundonsize$ by the first of the above equations.
  
  So we assume $\bar{t}(\ell) > \bar{s}(f(\ell))$ and $t(\ell) \le s(f(\ell))$.
  Thus, we have $\nu \le s(f(\ell)) + m_{\ell} + \bar{m}_\ell + \bar{t}(\ell)$.
  From the fourth of the above equations we have that
  $\bar{s}(f(\ell))  + s(f(\ell)) + m_{\ell} + \bar{m}_{\ell} -1 \le 2\bar{s}(f(\ell))  + 2 \bar{m}_{\ell}$
  and so
  $s(f(\ell)) + m_{\ell} + \bar{m}_{\ell} -1 \le \bar{s}(f(\ell))  + 2 \bar{m}_{\ell}.$
  We then combine with the upper bound on $\nu$ to obtain that
  $$\nu \le s(f(\ell)) + m_{\ell} + \bar{m}_{\ell} + \bar{t}(\ell) \le \bar{s}(f(\ell)) + \bar{t}(\ell) + 2 \bar{m}_{\ell} +1.$$
  Therefore, since $\bar{s}(f(\ell)) < \bar{t}(\ell)$, we conclude that
  $\nu \le 2\bar{t}(\ell) + 2 \bar{m}_{\ell}  \le \boundonsize$, using the third equation. The case where
  $\bar{t}(\ell) \le \bar{s}(f(\ell))$ and $t(\ell) \ge s(f(\ell))$ is symmetric.

  We now need to incorporate possibly one deleted client of $f(\ell)$. Namely a client served by $f(\ell)$ in $\opt$
  whose edge has been deleted during the preprocessing step and so that is still assigned to $f(\ell)$ (recall that at most
  1 client served by a facility is deleted during the preprocessing step).
  Observe that there is a deleted client only if $\eta$ is not a multiple of 2. In which case we have that
  $\boundonsize = \eta -1$ and so, the total number of clients assigned to $f(\ell)$ is $\boundonsize + 1 \le \eta$
  as claimed.

  We finally turn to Property~\ref{lem:intersol:cond3}. Since the assignment for client $c \in \calC-\calC_1$ is the same
  than in an optimal solution, the first sentence is clear. For the second part, we bound the distance
  from each client $c \in \calC_1$ to $\mu_c$ by the length of the path induced by the sequence $S(e)$, where $e$ is
  the edge associated to $c$. By the triangle inequality this indeed provides an upper bound on $\dist(c,\mu(c))$.
  We note here that this is a correct bound if the cost of a solution is the sum of distances (\ie for the $k$-median
  objective). In the case of the
  $\cost(c,\mu(c)) = \dist(c,\mu(c))^p$, with $p>1$, so we use the previous remark.

  Finally, to bound the sum of the length of the edges in all the paths associated to the route-to-matched sequences
  we simply invoke Lemma~\ref{lem:path}. 
  It follows, the total cost of the assignment $\mu$ for the vertices of $\calC_1$ is at most $4 (\cost(\opt)+ \cost(L))$,
  or using the above remark, at most $2^{O(p)} (\cost(\opt)+ \cost(L))$ for the case $p>1$.
\end{proof}

We can now prove the main proposition.

\begin{proof}[Proof of Proposition~\ref{prop:main}]
  %% Consider the set of unmatched facilities of $\opt$ and the route-to-unmatched sequences of these facilities.
  For each unmatched facility $f \in \opt$, we let $\xi(f)$ denote the unmatched facility of $L$ that is the closest to
  $\xi(f)$.
  We then divide the unmatched facilities of $\opt$ into two groups, $U_1,U_2$ as follows.
  Let $U_1 = \{ f \mid f \in \opt \text{ and there is no facility $f' \neq f$ s.t. $\xi(f') = \xi(f)$}\}$.
  Let $U_2$ be the rest of the unmatched facilities of $\opt$.

  We let $\tF = U_1 \cup \calM_A$ and let $\tL = \{\ell \mid \exists f \in U_1,~s.t.~ \xi(f) = \ell \} \cup \calM_B$,
  and $\phi$ be the 1-to-1 mapping
  of the facilities of $\tF$ to $\tL$ defined by the matching $\calM$ and the function $\xi$ on $U_1$. 

  We define $\hF = F - \tF = U_2$ and $\hL = L - \tL$. By the pigeonhole principle we immediately have
  that $|\hF| = |\hL|$ and $|\tF| = |\tL|$.
  To finish the proof of the proposition, we need to prove Properties~\ref{mainprop1} and~\ref{mainprop2}.

  We first aim at proving Property~\ref{mainprop1}. Consider the mapping $\xi$ of the facilities
  of $\hF$. We let $\chi(\ell) = \{f \mid \xi(f) = \ell\}$.
  We now proceed as follows: for each facility $\ell$ such that $|\chi(\ell)| > 1$, we pair up the facilities
  of $\opt$ such that $\xi(f) = \ell$. Let $(f_1,g_1),(f_2,g_2),\ldots$ be the list of $\lfloor |\chi(\ell)|/2 \rfloor$
  pairs. For each pair, we will consider closing one facility.
  We need to guarantee two things: first that the
  capacities are met and second that the total service cost is bounded.
  
  To ensure that the capacities are met, we make use of Lemma~\ref{lem:intersol}.
  For each pair $(f_i,g_i)$, we follow the assignment $\mu$ for the set of clients
  that they serve in $\opt$. Without loss of generality, assume that
  $f_i$ is farther away to $\ell = \xi(f_i) = \xi(g_i)$ than $g_i$. 
  This guarantees that both facilities serve at most
  $\lfloor \eta/2 \rfloor$ clients.
  We consider the cost of
  closing down $f_i$ and serving its clients by $g_i$.
  Moreover, $\mu$ reassigns clients
  served by the unmatched facilities to matched facilities
  and Lemma~\ref{lem:intersol} shows that the total number of clients assigned
  to a matched facilities is at most $\eta$. It follows that the total
  number of clients assigned to $f_i$ and $g_i$ is at most $\eta$ and so removing
  one of the two facilities still yield a feasible solution.

  We now turn to bounding the cost of closing one facility per pair $(f_i,g_i)$. 
  Consider first for simplicity the case $p=1$. The reassignment we have designed is as follows:
  \begin{enumerate}
  \item For the clients $c$ whose
    corresponding edge $e$ is s.t. $S(e)$ is a route-to-matched sequence
    and such that the facility serving $c$ belongs to a pair $(f_i,g_i)$,
    the assignment is the same as in $\mu$.
  \item For the clients $c$ whose
    corresponding edge $e$ is s.t. $S(e)$ is a route-to-unmatched sequence,
    and s.t. $c$ is served by a facility $f_i$ in a pair $(f_i,g_i)$,
    the assignment is now $g_i$. Let
    $\ell$ be the facility such that $\xi(f_i) = \ell$ and $\xi(g_i) = \ell$.
    The cost of the assignment is $\dist(c,g_i) \le \dist(c, \ell) + \dist(\ell, g_i)
    \le  \dist(c, \ell) + \dist(\ell, f_i) =  2\dist(c, \ell) + \opt(c)$, 
    by the triangle inequality and since $\dist(\ell, f_i) \ge \dist(\ell, g_i)$.
    We redefine $\mu(c) = g_i$.

  \item For the remaining clients, the assignment is the same than in $\opt$.
  \end{enumerate}

  Consider the clients $c$ whose
  corresponding edge $e$ is s.t. $S(e)$ is a route-to-matched sequence
  and such that the facility serving $c$ belongs to a pair $(f_i,g_i)$.
  Lemma~\ref{lem:intersol} shows that the sum over all pairs $(f_i,p_i)$ of the cost of
  the reassignment of their clients whose
  corresponding edge $e$ is s.t. $S(e)$ is a route-to-matched sequence is bounded by
  $O(\cost(\opt) + \cost(L))$.

  For the clients $c$ whose
  corresponding edge $e$ is s.t. $S(e)$ is a route-to-unmatched sequence,
  and s.t. $c$ is served by a facility $f_i$ in a pair $(f_i,g_i)$.
  Let $\ell$ be the facility such that $\xi(f_i) = \ell$ and $\xi(g_i) = \ell$.
  We have that the cost is $\dist(c,\mu(c)) \le 2\dist(c, \ell) + \opt(c)$. Now observe that
  $\dist(c,\ell) \le \text{length}(\calP(e))$ since $\ell$ is the closest
  unmatched facility to $f_i$.
  Thus applying Lemma~\ref{lem:path}, the sum over all the facilities
  $f_i$ that are closed of the reassignment cost of their clients $c$ 
  whose
  corresponding edge $e$ is s.t. $S(e)$ is a route-to-unmatched sequence
  is at most $O(\cost(L) + \cost(\opt))$.
  
  For the remaining clients, their cost is the same than in $\opt$.

  We thus have that:
  $\sum_{(f_i,g_i)} \sum_{c \text{ served by $f_i$ or $g_i$}} \dist(c,\mu(c))$
  is $O(\cost(\opt) + \cost(L))$.
  Moreover, $\mu(c)$ does not assign more than $\eta$ clients to any facility.

  Now consider selecting each pair $(f_i,g_i)$ with probability $\eps$
  and closing down $f_i$.
  For each selected pair, we follow the assignment prescribed above and
  for the remaining pairs, we follow the optimal assignment.
  The assignment is feasible no matter what are the selected pairs since we only consider reassigning clients
  served by the selected pairs in $\opt$ to matched facilities or
  to one of the facility of the pair. By Lemma~\ref{lem:intersol}, we know that we can reassign all clients
  of the pairs and still get a feasible solution, therefore the solution obtained is definitely
  feasible.
  
  %% It is easy to see that the assignment is feasible since there exists
  %% a feasible assignment where one facility per pair $(f_i,g_i)$ is closed.
  
  By the above discussion, the expected cost of the assignment
  for the clients that are served by a facility of a pair $(f_i,g_i)$,
  is at most
  \begin{align*}
  \sum_{(f_i,g_i)}~ (\text{pr}[(f_i,g_i)\text{ is selected}] \cdot\\
  \sum_{c \text{ served by $f_i$ or $g_i$}} \dist(c,\mu(c))) \\
  + \sum_{(f_i,g_i)}~(1-\text{pr}[(f_i,g_i)\text{ is selected}]) \\
  \sum_{c \text{ served by $f_i$ or $g_i$}} \opt(c)
  \end{align*}
  which is at most

  \begin{align*}
    \sum_{(f_i,g_i)}\sum_{c \text{ served by $f_i$ or $g_i$}}\opt(c) + \\
    O(\eps (\cost(\opt) + \cost(L))).
  \end{align*}
  
  Therefore, since for the rest of the clients, the cost is optimal,
  there exists a  solution $\globalSs$ of cost at most
  $\cost(\opt) +  O(\eps (\cost(\opt) + \cost(L))).$

  We finally prove Property~\ref{mainprop2} of the proposition.
  Consider a facility $f \in \tF$ and a facility $\ell \in \tL$ such that
  $\phi(f) = \ell$.
  Let $c(f)$ be the cost of replacing $f$ by $\ell$ in solution $\globalSs$
  and serving the set $N(\ell)$ of all the clients served by $\ell$ in $L$ by $\ell$ in the solution
  $\globalSs$.

  Our bound of $c(f)$ is in 2 steps. We first bound the cost of serving by $\ell$ all the clients assigned
  to $f$ in assignment $\mu$. This is an intermediate solution that does not satisfy that the clients
  served by $\ell$ in $L$ are also served by $\ell$ in solution $\globalSs$. We will then modify
  the intermediate solution to ensure this last property.

  %% We bound $c(f)$ by  serving the clients of $f$. We will count in that cost not only the cost of
  %% the client assigned to $f$ in the solution $\globalSs$ but all the clients assigned to $f$ via the
  %% assignment $\mu$. As we will see, this is an upper bound on the cost of replacing

  Consider first the case where $f$
  is an unmatched facility. We will reassign the clients of $f$ in two ways. First, for the clients
  $c$ whose corresponding edge $e$ is s.t. $S(e)$ is a route-to-matched sequence. In that case,
  we use $\mu(c)$ as a reassignment and so these clients are served by a different facility than $f$.
  Again, this is compatible with the previous reassignment since
  the mapping $\mu$ that reassigns all clients of unmatched facility ensures that no matched
  facility receives more than $\eta$ clients by Lemma~\ref{lem:intersol}.

  Second, for the set $N(f)$ of clients whose corresponding edge $e$ is s.t. $S(e)$ is a route-to-unmatched sequence,
  we temporarily assign them to $\ell$ and we can bound the cost for such clients by length($\calP(e)$),
  since $\ell$ is the closest unmatched facility.

  Now consider the case where $f$ is a matched facility. We proceed identically but we cannot use the bound
  on the length of the path since this bound only applies to unmatched facilities. In that case, we use the 
  bound given by the matching. We have that the cost paid by each client $c$ of $f$ is $\dist(c,f) + \dist(f,\phi(f))$.
  By the triangle inequality, $\dist(f,\phi(f))$ is at most the weight of the edge in the matching.
  Serving all the clients assigned to $f$ in mapping $\mu$ incurs an additional cost (in addition to what they
  are paying due in mapping $\mu$)
  $\eta \left( \sum_{\ell \in \calM_B} \dist(\ell,f(\ell))\right)$
  which is by Corollary~\ref{cor:matchingcost} at most
  $\eta \left( 2(\cost(\opt) + \cost(L))/\eta\right)$.

  Therefore, the reassignment performed for the intermediate solution has cost at most $O(\cost(L) + \cost(\opt))$.
  Note that this indeed takes into accound the reassignment of the client $c$ whose corresponding edge
  $e$ is s.t. $S(e)$ is a route-to-matched sequence.

  We now move from the intermediate solution to a solution where for each selected pair $(\ell,f)$, the clients
  served by $\ell$ in $L$ are also served by $\ell$ in $\opt$.  
  Let's assign the clients of $N(\ell)$ by $\ell$.
  By doing so, we may have exceeded the capacity of $\ell$: we have $|N(f)|, |N(\ell)| \le \eta$ but
  $|N(f)| + |N(\ell)| > \eta$. To fix this, we use the room left out on the other facilities
  by the $|N(\ell)|$ clients served by $\ell$ in solution $L$. Indeed, since these clients are now served by
  $\ell$, they leave some free room to the other clients.
  Thus, consider an arbitrary set $B$ of $|N(\ell)| + |N(f)|- \eta$
  clients of $|N(f)|$ (note that $|N(f)| \ge |N(\ell)| + |N(f)| - \eta$) and define a 1-to-1 mapping from this set to
  an arbitrary subset $B'$ of size
  $|N(\ell)| + |N(f)|- \eta$ of $N(\ell)$ (again note that $|N(\ell)| \ge |N(\ell)| + |N(f)| - \eta$ so this is possible).
  Now, each client of $B$ is assigned to the facility that serves the client it is mapped to in $B'$ in the solution
  $\globalSs$. Since $\globalSs$ is feasible, this solution is also feasible.
  By the triangle inequality, the increase in cost for doing so is at most $L(c) + \opt(c)$ where $c$ is the client
  in $B'$ and so, in total for the clients in $N(f)$ at most $\sum_{c \in N(\ell)} L(c) + \opt(c)$.

  Summing up over all such facilities and using Lemma~\ref{lem:path}, we have that
  $\sum_{f \in \tF} c(f) = O(\cost(L) + \cost(\opt))$. Thus, if each facility $f$ is replaced by $\phi(\ell)$ with probability $\xi^2$,
  we have that the expected cost of the solution is $\sum_{f \in \tF} pr[f \text{ selected by the random process}] \cdot  c(f) =
  \sum_{f \in \tF} \xi^2 c(f)  = O(\xi^2 \cost(L) + \cost(\opt))$.
  By Markov inequality, we have that the resulting solution has cost at most
  $(1+O(\xi))\cost(\opt) + O(\xi (\cost(L) + \cost(\opt)))$ as claimed.  
  This concludes the proof of the proposition.

  To handle the case $p>1$, one needs to proceed as prescribed in the previous remark.

\end{proof}

\section{A Simple QPTAS for Doubling Metrics
  -- Proof of Theorem~\ref{thm:qptas}}
\label{sec:qptas}
In this section, we give a simple approach for obtaining
an algorithm running in time $\exp(((d \rho \eps^{-1})^{1/\rho}
\log n)^{O(d)})$,
which is a quasi-polynomial bound for any fixed $d$.
In this section and the next, it will be convenient to see an
assignment of a client $c$ to a center $\ell$ as a path from
$c$ to $\ell$ that may intersect portals and whose length is simply
the $p$th power of the sum of the length of the segments of the path.

Our algorithm is very simple.
Let $\eps > 0$ be a sufficiently small constant.
Assume we know how to compute
a $\gamma$-approximate solution
$L$. We show how to compute a  solution of cost at most
$(1+\eps)\cost(\opt) + \eps \cdot \cost(L)$.
At start, the algorithm computes a randomized split-tree $\calD$.
Let's condition on
the event that $\calI_{\calD,L}$ is a good instance (w.r.t. $L$).
This happens with
probability at least $1-\eps$ by Lemma~\ref{lem:valid}.
The algorithm then computes $\calI_{\calD,L}$ and works in $\calI_{\calD,L}$,
its goal being to find the best valid solution in $\calI_{\calD,L}$.
We design a dynamic program that given $\calI_{\calD,L}$ and $L$
computes a $(1+\eps)$-approximation to the best valid solution. We then preprocess
this new instance
as follows. For each facility of $\ell$ that is badly cut, we
\emph{force it open}: our dynamic program is forced to pick $\ell$
in its solution. Since the dynamic program aims at finding
the best valid solution, we have that the best solution
in the preprocessed instance $\calI'_{\calD,L}$ can be transformed into
a valid solution in $\calI_{\calD,L}$ with the same cost.

The dynamic program proceeds on $\calD$ from the leaves to the
the root. The algorithm then computes a hierarchy of nets: namely
for each box $P$ of level $i$, it computes a
$\rho 2^{i+1}$-net which is a superset of the nets computed for the
descendant boxes of $P$.
The net of a box is then used as a set of \emph{portals}.
It follows that the number of portals at a given box is
$\rho^{-O(d)}$.

The definition a cell $C$ of the dynamic program is a tuple
\begin{align*}
  (B, \langle (n^{\text{in}}_{p_1, d_1}, n^{\text{in}}_{p_1, d_2}, \ldots, n^{\text{out}}_{p_1, d_{\max}}),\\
(n^{\text{in}}_{p_2,d_1}, \ldots, n^{\text{out}}_{p_2, d_{\max}}),\ldots,\\
   (n^{\text{in}}_{p_{\rho^{-O(d)}}, d_1},\ldots,n^{\text{out}}_{p_{\rho^{-O(d)}}, d_{\max}})\rangle, k_B)
\end{align*}
where $p_1,\ldots,p_{\rho^{-O(d)}}$ are the portals of box $B$, and $d_i$ are power of $(1+\eps^2/\log n)$
in the range [minimum distance; maximum distance]. Given such a cell $C$ we say that 
$n^{\text{out}}_{p_{i}, d_{j}}$ is a parameter of $C$.

Then, the value of such a table cell is the 
the value of the best valid solution for the clients
in box $B$ under the constraints that:
\begin{enumerate}
\item For each portal $p$ of $B$,
  $n^{\text{in}}_{p_i, d_j}$ 
  clients from the inside of $B$ at distance in $[d_j;(1+\eps)d_j]$ that are served outside $B$
  and crossing at $p$, $n^{\text{out}}_{p_i, d_j}$ 
  clients from the outside of $B$ at distance in $[d_j;(1+\eps)d_j]$ that are served inside $B$
  and crossing at $p$. Note that only one of the two can be positive;
\item The solution opens $k_B$ of centers open inside $B$,
  including the badly cut
  centers of $L$ inside $B$;
%% The minimum cost of serving all clients inside $B$ using any set of
%%   $k_B$ centers in $B$ plus the portals, where each portal has capacity
%%   $n_p$.
\end{enumerate}

Eventually, we consider the solutions at the root $R$, with the
following set of parameters,
each portal $p$ of $R$ is such that $n_p = 0$, and $k_R = k$. Among all
these solutions, the algorithm output the one with minimum cost.

The base cases of the dynamic program consist of smallest-size cells where at most
one facility can be opened at a given location. The value of each base-case cell
can be obtained by computing the assignment of each $n^{\text{out}}_{p_i, d_j}$ to
the open facility (if there is one facility open, otherwise each $n^{\text{out}}_{p_i, d_j} = 0$),
namely $n^{\text{out}}= \sum_{i} \sum_j n^{\text{out}}_{p_i, d_j} \cdot d_j^p$. In addition, the number of clients
within the cells that are assigned outside should be at most $n^{\text{in}} = \sum_i n^{\text{in}}_{p_i, 0}$ for each portal
$p_i$, and the total number of clients in the cell minus $n^{\text{in}}$ plus $n^{\text{out}}$ should be
at most $\capac$.

Computing the value of a non-base-case cell  $C$ 
\begin{align*}
  (B, \langle (n^{\text{in}}_{p_1, d_1}, n^{\text{in}}_{p_1, d_2}, \ldots, n^{\text{out}}_{p_1, d_{\max}}),\\
  (n^{\text{in}}_{p_2,d_1}, \ldots, n^{\text{out}}_{p_2, d_{\max}}),\ldots,\\
  (n^{\text{in}}_{p_{\rho^{-O(d)}}, d_1},\ldots,n^{\text{out}}_{p_{\rho^{-O(d)}}, d_{\max}})\rangle, k_B),
\end{align*}
of the DP can be done by iterating over all tuples of cells of the DP such that 
each child box $B'$ of $B$ appears in exactly one DP-cell of the tuple, and all DP cells
of the tuples are associated with child boxes of $B$, and taking the tuple that is compatible
and whose sum of values of entry cells is minimized.
The tuple is compatible with cell $C$ if for each $n^{\text{in}}_{p_i,d_1j}$ of the definition
of $C$,
one can assign $n^{\text{in}}_{p_i,d_j}$ clients to each portal
$p_i$ under the constraint that one can assign at most $n^{\text{in}}_{p_u,d_v}$ from a
from each portal $p_u$ of a DP-cell $C'$ of the tuple such that $n^{\text{in}}_{p_u,d_v}$ is
part of the definition of $C'$ and  $d_v + \dist(p_u,p_i) = (1\pm \eps^2/\log n) d_j$.
Moreover the sum of the $k_{B'}$ for each $B'$ of the tuple has to be at most $k_B$.
The verification of the compatibility of an assignment is done through enumeration of all
possibilities.

We show that our dynamic program outputs a solution of cost
at most $(1+\eps) \cost_{\calI_{\calD,L}} \cost(\globalSh)$, where
$\globalSh$ is as defined per the definition of valid solution.

\begin{lemma}
  \label{lem:dplargedim}
  The above dynamic program produces a solution of cost at
  most $(1+\eps) \cost_{\calI_{\calD,L}} \cost(\globalSh)$, where
  $\globalSh$ is as per the definition of valid solution. Moreover, the
  running time is at most $\exp((\rho/\eps)^{-O(d)})$.
\end{lemma}
\begin{proof}%[Proof of~Lemma~\ref{lem:dplargedim}]
  To prove the approximation gurantee, we need to argue that:
  \begin{enumerate}
    \item Forcing the assignment path of each client to make a
      detour through the closest portal whenever it leaves a region
      in the new instance does not increase the cost
      by a factor of more than $(1+\eps)$; and
    \item Rounding the number of clients  coming through each portals to a power
      of $(1+\eps^2/\log n)$ yields a solution of cost at most $(1+\eps)
      \cost_{\calI_{\calD,L}} \cost(\globalSh)$.
  \end{enumerate}

  We first prove (1).
  Consider instance $\calI_{\calD,L}$. 
  Consider a client $c$ that is not badly cut
  and $\calI'_{\calD,L}$ together with
  the facility $\globalSh(c)$ that serves it in solution $\globalSh$.
  We have that $c$ and $\globalSh(c)$ are separated at a level $u(c)
  \le \log (d  \frac{\log n}{\epsval}) + \log \dist(c,\globalSh(c))+1$,
%% such that 
%% of at most $r(c) = 2(d \log n /\epsval) \log (d(c,\globalSh(c)))$,
  by the definition
  of (not) badly cut.
  Thus, consider the solution where, in each level where
  $c$ and $\globalSh(c)$ are
  in different boxes, the path makes a detour to the closest portal
  in the box containing $c$. This incurs a detour of
  $\rho 2^{i+1}$ for each such box of level $i$.
  We have that the total detour for the path is at most
  \begin{align}
    \label{eq:portals}
    \sum_{i=0}^{u(c)} \rho 2^{i+1}
    &\le \sum_{i=0}^{u(c)} \rho 2^{u(c)+1 -i} = \sum_{i=0}^{u(c)} \rho
    \frac{2^{u(c)+1}}{2^i} \\
    &\le \rho 2^{u(c)+2} \nonumber \\
    &\le 16 \rho
    (d \frac{\log n}{\epsval}) \cdot  \dist(c,\globalSh(c))
    \nonumber
  \end{align}

  Setting $\rho = \eps (16d \frac{\log n}{\epsval})^{-1}$ shows
  that the overall detour is at most $\eps \dist(c,\globalSh(c))$.

  Now consider a client $c$ that is badly cut. In instance
  $\calI_{\calD,L}$, we have that $c$ is relocated to the facility $f$ of
  $L$ that serves it in solution $L$. Then two things can happen: if
  $f$ is not badly cut, then in instance $\calI_{\calD,L}$ we have that $c$
  is $c$ and the facility $\globalSh(c)$ that serves it in $\globalSh$
  are separated at a level $u(c)
  \le \log (d  \frac{\log n}{\epsval}) + \log \dist(c,\globalSh(c))+1$.
  Hence, the cost increase for $c$ due to
  enforcing the detours to the closest portals in instance $\calI_{\calD,L}$
  is at most $(1+\eps)$ times the assignment cost of $c$ in solution
  $\globalSh(c)$.
  Otherwise, $f$ is badly cut and so, by definition of $\globalSh$,
  $f$ serves $c$ in $\globalSh$ and so the portals have no effect.

  An immediate induction shows that our dynamic program computes 
  a valid solution $S'$ such that 
  $\cost_{\calI_{\calD,L}}(S') \le
  (1+O(\eps))\cost_{\calI_{\calD,L}}(\globalSh)$. Namely, the cost
  incurred by rounding distances to power of $(1+\eps^2/log n)$ over
  the $O(\log n)$ levels of the dynamic program only incurs an overall
  cost increase of $(1+\eps^2/\log n)^{\log n} = (1+O(\eps^2))$. The
  rest of the dynamic program is exact.
  Thus, Lemma~\ref{lem:valid} implies
  that $\cost(S') \le (1+\eps)\cost(\opt) + \eps\cost(L)$.

  The running time follows from the definition of the dynamic program
  and the choice of the number of portals and the fact that
  $\frac{\text{max distance}}{\text{min distance}} = n^{O(1)}$. The running
  time is
  $\exp((\rho/\eps)^{-O(d)})$.
\end{proof}

The proof of Theorem~\ref{thm:qptas} almost follows from combining
Lemma~\ref{lem:dplargedim} and Lemma~\ref{lem:valid}.
The only thing that remains to be proven is that we can
find a solution $L$
of cost $O(\cost(\opt))$. Unfortunately, nothing better than
an $O(\log n)$-approximation is known. We thus repeat the above
algorithm until we find a solution of cost at most $(1+\eps)\cost(\opt)$.
Namely, we start with $L$ being an $O(\log n)$-approximation.
We then apply the algorithm to obtain a solution $L_1$ of cost
at most $(1+\eps)\cost(\opt) + \eps\cost(L)$ with probability
at least $1-\eps/\log n$; boosting the probability is always possible
by repeating the random step (i.e.: computing a new
randomized split-tree
decomposition) and outputing the best solution among the
different solution computed.
We then apply the algorithm again to $L_1$ and find a solution $L_2$
of cost at most  $(1+\eps)\cost(\opt) + \eps\cost(L_1)$ with probability
at least $1-\eps/\log n$. Repeating this process $s =
O(\log \log n)$ times
yields a solution $L_{s}$
of cost at most $(1+\eps)\cost(\opt)$ with
probability at least $1-\eps$.
We note that similar techniques have been used in~\cite{BateniDHM16}.

%% \frac{2 d \log n}{\epsval} \cdot
%% \frac{d(c,\globalSh(c))}{2^i} \cdot \frac{1}{\rho}
%% = O(\frac{2 d \log n}{\rho \epsval} d(c,\globalSh(c)))
%% $$

\section{A PTAS for Capacitated $k$-Median in $\R^2$ -- Proof of Theorem~\ref{thm:ptas}}
\label{sec:ptas}
We aim at providing a faster algorithm for the Euclidean plane
with running time $2^{\rho^{-1}} \text{poly}(n)$, which
as we will see is $n^{\eps^{-O(1)}}$.

%% For $k$-median inputs in $\R^2$, we can improve the running time of
Our algorithm differs from the 
the algorithm described in Section~\ref{sec:qptas} in two ways:
\begin{enumerate}
\item\label{step:dp1}
  the size of the net in each box is $\rho^{-1} =f(\eps) \log n$
  for some computable function $f$.
\item\label{step:dp2} the table entries of the dynamic program
  are smaller: the number $n_p \in [n]$ that is kept for each
  portal $p$ is a power of $(1+\eps)$.
\end{enumerate}

In the case of the plane, the
algorithm uses the randomized quad-tree dissection of Arora.
This also satisfies the properties 1-4 described in Section~\ref{sec:badly}.
Any box of the dissection of level $i$ is a $2^{i+1} \times 2^{i+1}$ square.
We condition on the event that the instance $\calI_{\calD,L}$ is good
which by Lemma~\ref{lem:valid} happens with probability at least
$1-\eps$.

Reducing the size of the net simply follows the standard analysis of Arora et al.~\cite{ARR98}.
We only place portals on the boundaries of boxes and consider solutions that are such that
each path connecting a client $c$ to the facility $f$ that serves it in $\opt$ is forced
to make a detour for each boundary of a box $B$ it crosses.
When using a portal set of size $\rho$, the length of each such
detour is then $\rho \partial B$, where $\partial B$ denotes the perimeter of box $B$ and so less than twice is diameter.

Therefore, the discussion of Section~\ref{sec:qptas} still holds:
from Lemma~\ref{lem:valid} and the definition of a valid instance,
the badly cut clients are not a problem anymore: their
assignment cost is
either 0 (in case their center in $L$ is also badly cut) or
well approximated by portals (c.f. Eq~\ref{eq:portals}). Similarly
for the clients that are not badly cut, their assignment cost
is well approximated by portal (c.f. Eq~\ref{eq:portals} again).
More formally,
let $B(c)$ denotes the set of box that contains $c$ and that do not
contain the facility $f$
that serves $c$ in $\opt$.
Thus, we write that the total detour for using portals
is at most $\sum_{B \in B(c)} \rho 2^{\ell(B(c))+3}$, where $\ell(B(c))$
is the level of box $B(c)$. 
Therefore, the overall detour is at most $16 \eps \dist(c,f)$
since we condition on the event that
$\calI_{\calD,L}$ is a good instance.

This implies an algorithm running in time $\exp(O(f(\eps) \log^2 n))$
through the
following naive implementation. The dynamic program
defines an entry by a tuple
$(B, \langle (n^{\text{in}}_{p_1},n^{\text{out}}_{p_1}),
(n^{\text{in}}_{p_2},n^{\text{out}}_{p_2}), \ldots (n^{\text{in}}_{p_{\rho^{-1}}},n^{\text{out}}_{p_\rho^{-1}})\rangle, k_B)$, where $p_1,\ldots,p_{\rho^{-1}}$ are the
portals of box $B$.
Then, the value of such a table entry is the 
the value of the best valid solution for the clients
in box $B$ under the constraints that:
\begin{enumerate}
\item for each portal $p$ of box $B$, $n^{\text{in}}_p$ clients 
  coming from the inside of $B$ and that are assigned to a
  center outside $B$ and crossing at $p$, and $n^{\text{out}}_p$ clients
  outside $B$ are
  assigned to a center inside $B$\footnote{Note that in an optimal
    assignment at most one of $n^{\text{in}}_p,n^{\text{out}}_p$ is
    non-negative but this does not impact the asymptotic running time.}
  and are crossing at $p$.
\item the solution opens at most $k_B$ centers inside $B$.
\end{enumerate}
%% The value of a DP cell is the best portal-respecting solution that places $k_B$
%% centers inside $B$, while assigning $n^{\text{in}}_p$ clients of $B$ to each portal
%% and serving $n^{\text{out}}_p$ clients from each portal $p$.
The running time of this algorithm is $\exp{O(\rho^{-1}\log n)}$.
However, we aim at getting a truly polynomial-time algorithm. We thus
makes the following optimization.

We now show how to speed this up. For each box $B$, we pick
an arbitrary portal of $p$, $p^*_B$ -- when $B$ is clear
from the context we simply refer to it by $p^*$.
Since portals are
placed on the perimeter of a square, we can impose an ordering from $p^*$
in clockwise manner: $p^* = p_0$, $p_1$ is the next portal
on the boundary after $p^*$ in the clockwise order, and so on.

Our algorithm only considers
table entries $(B, \langle (n^{\text{in}}_{p_1},n^{\text{out}}_{p_1}),
(n^{\text{in}}_{p_2},n^{\text{out}}_{p_2}), \ldots
(n^{\text{in}}_{p_{\rho^{-1}}},n^{\text{out}}_{p_\rho^{-1}})\rangle, k_B)$
that are such that
$\eps^2 n_{p_i} \le n_{p_{i+1}} \le n_{p_i}/\eps^2$. We say that such
a table entry satisfies constraint 1. A solution satisfies constraint
1 if it can be described by a table entry that satisfies constraint 1
for each box $B$.

Moreover, our algorithm only considers
table entries $(B, \langle (n^{\text{in}}_{p_1},n^{\text{out}}_{p_1}),
(n^{\text{in}}_{p_2},n^{\text{out}}_{p_2}), \ldots
(n^{\text{in}}_{p_{\rho^{-1}}},n^{\text{out}}_{p_\rho^{-1}})\rangle, k_B)$
that satisfy constraint 1 and where each $n^{\text{in}}_{p}$ and
$n^{\text{out}}_p$ is a multiple of $(1+\eps_0)$, for some
$\eps_0$ that will be chosen later, except for
$p^*_B$ for which the value is still
in $[n]$.  We refer to this constraint imposed
on solutions as constraint 2: namely, a solution that can be described
through by such table entries is a solution that satisfies constraint 2.
Similarly, a table entry is said to satisfy constraint 2 if each of the
$(n^{\text{in}}_{p_i},n^{\text{out}}_{p_i})$ are power of $(1+\eps_0)$,
except for $p^*$.

\begin{claim}
  The running time of the dynamic programming algorithm
  that only considers entries that satisfy constraint 2 (and so also
  constraint 1) is $ n^{O(1)}\exp(\rho^{-1}\log((\eps\eps_0)^{-1}))$.
\end{claim}
\begin{proof}
  Fix a given box $B$. There are $n^2$ possible values of
  $(n^{\text{in}}_{p^*},n^{\text{out}}_{p^*})$ and $k$ possible values for
  $k_B$. Then, for the $(n^{\text{in}}_{p_i},n^{\text{out}}_{p_i})$, each
  value is in the range $\eps^2 n_{p_{i-1}} \le n_{p_{i}}
  \le n_{p_{i-1}}/\eps^2$ and a power of $(1+\eps_0)$ and so there
  are at most $O(\log(1/\eps)/\eps_0)$ possibilities.
  This shows that the total number of entries
  $(B, \langle (n^{\text{in}}_{p_1},n^{\text{out}}_{p_1}),
  (n^{\text{in}}_{p_2},n^{\text{out}}_{p_2}), \ldots
  (n^{\text{in}}_{p_{\rho^{-1}}},n^{\text{out}}_{p_\rho^{-1}})\rangle, k_B)$
  that satisfy constraint 2 is at most
  $(\log(1/\eps)/\eps_0)^{\rho^{-1}} n^3$.
  It follows that the running time of the algorithm is at most
  $\exp{\rho^{-1} \log((\eps\eps_0)^{-1})} n^{O(1)}$.
\end{proof}

To conclude, we now need to prove that there exists a near-optimal
solution that satisfies constraint 2, this is done in the following
claim. The proof of the theorem then follows immediately.

\begin{claim}
  \label{claim:r2}
  If $\eps_0 \le \eps^5$, 
  there exists a $(1+O(\eps))$-approximate
  solution that satisfies constraint 2 and that
  forces the assignment path of
  each client to leave a box only at its portals.
\end{claim}
\begin{proof}
  We first observe that our analysis of the detours paid by each client
  to reach the facility it is assigned to in an optimal solution
  is as follows.
  Each client $c$ going through a portal of a box $B$ of level $i$
  pays a detour of up to $\alpha \rho 2^i$ in the worst-case,
  where $\alpha$ is a large enough constant. In the case where each
  client pays a worst-case detour, the cost of the best solution
  that forces the assignment path of
  each client to leave a box only at its portals remains
  at most $(1+\eps)\opt$ (for the right choice of $\rho$, as
  a function of $\alpha$).

  Thus, we will give each client $c$ an additional \emph{budget} of
  $\alpha \rho 2^i$ for level $i$. The best solution 
  that forces the assignment path of each
  client to leave a box only at its portals plus
  pays the budget is at most $(1+2\eps)\opt$.
  Namely, any $(1+\eps)$-approximate
  solution such that each client pays its budget
  remains a $(1+O(\eps))$-approximation to an
  optimal solution.
%%   We use this remaining slack to fasten the algorithm.
%%   The idea is to restrict the values that $n^{\text{in}_p}$  and $n^{\text{out}_p}$ can
%%   take. We show the transformation for $n^{\text{in}_p}$, it works identically for 
%% $n^{\text{out}_p}$.

  We first consider the cost of making a solution satisfy constraint 1.
  We first claim that for two portals $p,p'$
  that are consecutive on the boundary of a box $B$ of level $x$,
  there exists a near-optimal
  solution such that $\eps^2 n_{p'} \le n_p \le n_{p'}/\eps^2$
  (and each client its assignment path leaves a box only at its portals).
  Indeed, consider the optimal solution that
  forces the assignment path of
  each client to leave a box only at its portals,
  and the number of clients
  assigned to $n_p$ and $n_{p'}$
  and assume that $\eps^2 n_{p'} > n_p$. Then, consider
  forcing $\eps^2 n_{p'}$ clients that go to $n_{p'}$ to make an
  extra-detour to
  $n_p$. The length of this detour is at most $\rho 2^{x+1}$.
  Hence, the extra-cost is
  $n_{p'} \rho 2^{x+1}$. Now, observe that this is
  at most $\eps^2$ times the budget
  of all the clients going to $n_{p'}$. A similar argument holds
  in the case where
  $n_p > n_{p'}/\eps^2$. Thus, consider the boundary of $B$ and
  an arbitrary portal
  $p_0$ of $B$. Let $p_0,p_1,\ldots$ be the portals in the order given
  by a clockwise
  walk on the boundary of $B$ starting at $p_0$.
Visit the portals in that order and ensure
that $\eps^2 p_{i}  \le p_{i+1}$ for all $i$ using the above transformation iteratively.
What is the overall cost? observe that the clients $n(p_i)$
that are initially going through portal
$p_i$ may now be assigned to a portal $p_j$ where $j$ is much larger than $i$.
What is the total cost for this? We have that at each portal at most an $\eps$ fraction
of the clients can be moved again. Thus, the total extra cost for the clients of $n(p_i)$ is
at most $n_{p_i} \eps^{2(j-i)} (j-i) \rho 2^{x+1}$.
Summing up over all portals, this yields a geometric sum and so is
at most
$O(n_{p_i} \eps \rho 2^{x+1})$ which is less than the sum of the budgets at level $i$ of the clients
going at $n(p_i)$ (for a large enough choice of $\alpha$).
Then visit the portals in the reverse order and proceed identically to ensure
that $p_{i}/\eps^2   \le p_{i+1}$. This choices that
there exists a
near-optimal solution that forces the assignment path of each
client to leave a box only at its portals and that satisfies constraint
1.

We now turn to prove that given such a solution, there exists a
solution that forces the assignment path of each
client to leave a box only at its portals and that satisfies constraint
2.
We show that that except for one portal denoted by $p^*$,
the numbers $n_p$ could be approximated to power of $(1+\eps^5)$
in the following way.
We again consider the portals in clockwise order, starting from $p_0 = p^*$.
The initial number of clients $n^0_{p_i}$ assigned to portal $p_i$ is the one prescribed
after the above transformation.
For the $i$th portal $p_i$, $i>0$, let $n_{p_i}$ number of clients assigned to
$p_i$ when the procedure visits $p_i$. Let $\widetilde{n_{p_i}}$ be the power of $(1+\eps^5)$
that is the closest to $n_{p_i}$ and smaller than  $n_{p_i}$. We reassign $n_{p_i} - \widetilde{n_{p_i}} \le
\eps^5 n_{p_i}$ clients of $n(p_i)$ to $p_{i+1}$.

By doing so iteratively, we end up with an assignment where, except for $p^*$
which may receive from $p_{\rho-1}$ and not give to any other portal,
$n_{p_i}$ is a power of $(1+\eps^5)$.
We now bound the cost of the reassignment. %% Let $n^0_{p_i}$ be the number
%% of clients assigned to portal $p_i$ initially, namely after the previous reassignment.
We first show that $n_{p_i} \le (1+\eps^2) n^0_{p_i}$. This is true for $i\in\{1,2\}$ since
$n^0_{p_1} \le n^0_{p_2}/\eps$ and the total number of clients moved from $p_1$ to $p_2$
is at most $\eps^5 n^0_{p_1}$ and so $n_{p_2} \le (1+\eps^3) n^0_{p_2}$.
We assume that this is true up to $p_{i-1}$ and show that it holds for $p_i$.
The number of clients received by $p_i$ is thus at most $\eps^5 (1+\eps^2) n^0_{p_{i-1}}$
by the inductive hypothesis. This is at most $\eps^5 (1+\eps^2) n^0_{p_{i}}/\eps^2
\le \eps^3 (1+\eps^2) n^0_{p_{i}} \le \eps^2  n^0_{p_{i}}$ for any $\eps \le 1/2$.

It follows that the clients of $n(p_i)$ that are reassigned to $p_{i+1}$ can be chosen
from the clients that are initially assigned to $p_i$ and so each client that
is assigned
to portal $p_i$ (in the solution satisfying constraint 1)
is now assigned either to portal $p_i$ or to portal $p_{i+1}$.
It follows that the extra cost is at most the total budget of level $i$
for the clients going to the portals and the claim follows.
\end{proof}

\bibliographystyle{abbrv}
\bibliography{facilitylocationptas.bib}
\end{document}